\documentclass[11pt,twocolumn]{IEEEtran}
\usepackage{amsmath,amssymb,verbatim}
\usepackage[dvips]{graphicx}

\DeclareMathOperator{\Gal}{Gal}
\DeclareMathOperator{\ind}{ind}

\newtheorem{definition}{Definition}[section]
\newtheorem{thm}{Theorem}[section]
\newtheorem{proposition}[thm]{Proposition}
\newtheorem{lemma}[thm]{Lemma}
\newtheorem{corollary}[thm]{Corollary}
\newtheorem{exam}{Example}[section]

\newtheorem{remark}{Remark}[section]

\newcommand{\CC}{\mathbb{C}}
\newcommand{\ZZ}{\mathbb{Z}}
\newcommand{\RR}{\mathbb{R}}
\newcommand{\QQ}{\mathbb{Q}}

\newcommand{\Ac}{\mathcal{A}}
\newcommand{\Cc}{\mathcal{C}}
\newcommand{\Dc}{\mathcal{D}}
\newcommand{\Gc}{\mathcal{G}}
\newcommand{\Oc}{\mathcal{O}}
\newcommand{\Nc}{\mathcal{N}}

\newcommand{\bv}{\mathbf{b}}
\newcommand{\ev}{\mathbf{e}}
\newcommand{\gv}{\mathbf{g}}
\newcommand{\xv}{\mathbf{x}}
\newcommand{\yv}{\mathbf{y}}

\newcommand{\Hv}{\mathbf{H}}

\newcommand{\proj}{{\rm proj}}
\newcommand{\tr}{{\rm Tr}}
\newcommand{\diag}{{\rm diag}}

\newcommand{\Z}{{\mathbb Z}}
\newcommand{\R}{{\mathbb R}}
\newcommand{\C}{{\mathbb C}}
\newcommand{\D}{{\mathcal D}}

\newcommand{\HH}{{\mathbf H}}
\newcommand{\OO}{{\mathcal O}}

\newcommand{\A}{{\mathcal A}}

\newcommand{\Q}{{\mathbb Q}}

\newcommand{\mindet}[1]{\hbox{\rm det}_{min}\left( #1\right)}

\begin{document}

\title{Fast-Decodable Asymmetric Space-Time Codes  from Division Algebras}
\author{Roope Vehkalahti, \emph{Member, IEEE}, Camilla Hollanti, \emph{Member, IEEE}, and Fr\'ed\'erique Oggier
\thanks{Part of this work appeared at ISIT 2010 \cite{OVH_isit10},  
at SPCOM 2010 \cite{spcom}, and at ISITA 2010 \cite{VHL_isita10}.}
%\IEEEcompsocitemizethanks{
%\IEEEcompsocthanksitem}
}\maketitle

%*********************************************************************%
%
% ABSTRACT
%
%*********************************************************************%

\begin{abstract}
Multiple-input double-output (MIDO) codes are important in the near-future wireless
communications, where the portable end-user device is physically small and will typically  contain at most two receive antennas. Especially tempting is the $4\times 2$ channel due to its immediate applicability in the digital video broadcasting (DVB). Such channels optimally employ rate-two space-time (ST) codes consisting of $(4\times 4)$ matrices. Unfortunately, such codes are in general very complex to decode, hence setting forth a call for constructions with reduced complexity.

Recently, some reduced complexity constructions have been proposed, but they have mainly been based on different \emph{ad hoc} methods and have resulted in isolated examples rather than in a more general class of codes. In this paper, it will be shown that a family of  division algebra  based MIDO codes  will always result in at least 37.5\% worst-case complexity reduction, while maintaining full diversity and, for the first time, the non-vanishing determinant (NVD) property. The reduction follows from the fact that, similarly to the Alamouti code, the codes will be subsets of matrix rings of the Hamiltonian quaternions, hence allowing simplified decoding.  At the moment, such reductions are among the best known for rate-two MIDO codes \cite{BHV,SR-fast}.  Several explicit constructions are presented and shown to have excellent performance through computer simulations.

\end{abstract}


%*********************************************************************%
%
% INTRODUCTION
%
%*********************************************************************%

\section{Introduction}

Among known space-time codes, the  Alamouti code \cite{Alam} and the fully diverse $4\times1$ quasi-orthogonal codes \cite{HLL} stand out due to their orthogonality properties that are beneficial for  decoding. Both of these codes  however have a low code rate, hence best suitable for an asymmetric transmission, where there are less receive antennas than transmit antennas. It is far from obvious how to generalize these codes to asymmetric scenarios where we demand higher code rates and different number of  antennas. On the other hand, the now well known cyclic division algebra (CDA) codes designed for a symmetric transmission have  full rate and are generalizable to an arbitrary number of antennas. Unfortunately, they are very complex to decode, especially when we have less receive antennas than transmit antennas. Yet there is a strong demand for asymmetric codes that would be  fast-decodable, generalizable to more antennas, and would support higher rates. The special case of two receive antennas is referred to as a multiple input-double output (MIDO) code.

For example one of the most interesting wireless applications currently is the design of $4\times 2$ MIDO codes. Such asymmetric systems can be used in the communication between, for instance, a TV broadcasting station and a portable digital TV device. The four transmitters can either be all at one station or separated between two different stations in this way providing better coverage in the case when the transmission of one of the stations is blocked out by a deep shadow. 

In Europe, the digital video broadcasting (DVB) consortium has adopted different standards for terrestrial (DVB-T) fixed reception, handheld (DVB-H) reception, satellite (DVB-S) reception as well as an hybrid reception like DVB-SH. The ongoing work towards the standardization of the DVB Next Generation Handheld (NHG, see the DVB Project's web page \cite{DVB} for more information)  systems is bringing this topic ever more to the forefront of current MIMO research. The inclusion of the $4\times 2$ systems in the consortium's call for technologies for the DVB-NGH indicates having a MIDO code in the coming  standard.

One solution to the $4\times 2$ code construction problem could be to use a full-rate CDA code, \emph{e.g.} the $4\times 4$ Perfect code \cite{BORV}. However, when received with two antennas, a rate-four code cannot be optimally decoded with  a linear decoder such as a sphere decoder. 
Codes especially designed for the $4\times 2$ channel have been proposed in \emph{e.g.} \cite{asykonstru, asyoptimal, HHVMM}, but all the codes require high complexity
maximum-likelihood (ML) decoding, namely full-dimensional sphere decoding.

A natural approach to this design problem  is to imitate the form of the code matrices of the already known fast-decodable codes  or use these codes as building blocks for higher rate codes. The key problem in such constructions is that it is very hard to guarantee that the resulting code will still have good performance, thus in many cases requiring optimization to be carried out through extensive computer searches.

In this paper we are going to adopt a different approach to this problem.  We  study the algebraic structure of known fast-decodable  codes like the Alamouti code and the division algebra based quasi-orthogonal codes. By analyzing the relation between the Hasse-invariants and the geometric structure of these codes we are able to distill the key algebraic properties  that force these codes to be fast-decodable. This approach then depicts an infinite family of fast-decodable codes from division algebras. 

The main advantage of our take on this subject is that  the proposed codes are based on orders of division algebras and therefore they are not only fast-decodable, but are also  guaranteed to have full-diversity,  the non-vanishing determinant (NVD) property, and further  allow us to perform algebraic  minimum determinant optimization.  We can show, under  given conditions, that the ML decoding complexity of a MIDO code will always be reduced by at least 37.5\%, while maintaining the NVD.  Explicit constructions  based on the proposed criteria will be provided. One of the examples  introduces a code that has comparable performance with the best known fast-decodable ST codes \cite{BHV,SR-fast} and further has (provable) NVD.  The proposed theory provides fully diverse, fast-decodable (FD) codes with the NVD property for any even number $n_t$ of Tx antennas and any code rate $\leq n_t/2$. Motivated by the DVB-NGH, most of the examples are given in the case of 4 Tx antennas and 2 Rx antennas. 

We make the typical assumption of transmission over a  coherent i.i.d. Rayleigh fading channel with perfect channel state information at the receiver (CSIR) and with no CSIT,
$$
Y=HX+N,
$$
where $Y,X,H,N$ are the received, transmitted, channel, and the Gaussian noise matrix, respectively. The ST matrix $X\in M_{n_t}(\mathbb{C})$, while $Y,H,N\in M_{n_r\times n_t}(\mathbb{C})$, where $n_t$ (resp. $n_r$) denotes the number of transmit  (resp. receive) antennas.  We assume no correlation, but in the correlated case the transmitter can adapt to the rate-one code naturally embedded within the proposed codes while maintaining and even improving  fast decodability.

%*********************************************************************%
\subsection{Related work}

The first reduced ML-complexity $4\times 2$ construction was given in \cite{BHV}, combining two copies of a quasi-orthogonal code \cite{TJC}.  This resulted in a MIDO code that does have lower decoding complexity, but unfortunately does not have full rank. Nevertheless,  good performance is still achieved at low-to-moderate SNRs and with four real dimensions less in the sphere decoder.

The most recent results on fast-decodable codes have appeared in
\cite{SR-fast}, where new constructions with optimized performance have been presented, and in
\cite{OVH_isit10,spcom,VHL_isita10}, where fast-decodable codes with the NVD property have been built from crossed product and cyclic presentations of division
algebras. In the preprint \cite{JR} the authors consider quadratic forms as a tool for characterizing the decoding complexity, and in the preprint \cite{RaRa} multi-group ML-decodable collocated and
distributed space-time codes are proposed.

%*********************************************************************%
\subsection{Organization and contributions}

The rest of the paper is organized as follows.
We start by giving some background on space-time codes with a lattice
structure and their decoding via sphere decoding in Section \ref{sec:STC}.
The concept of fast decodability is then defined and illustrated in
Section \ref{sec:FD}, where the role of the Alamouti code is emphasized.
To pursue the study of fast-decodable codes, we then focus on CDA codes in Section \ref{sec:cda}, where some background and further
motivating examples are presented, translating fast decodability into
being able to embed the considered cyclic algebra into an algebra of matrices
with quaternionic coefficients. The conditions guaranteeing the existence
of such an embedding are studied in Section \ref{embed}: we need an algebra
whose center is totally real and such that all its infinite places ramify in
the algebra. A family of such cyclic algebras is provided. A last design criterion, 
the normalized minimum determinant, is added and bounds on optimal 
lattice codes with respect to it are computed in Section \ref{bounds}. Different explicit construction methods are described in Section \ref{explicit}. Finally,
several code constructions are presented in Section \ref{sec:codes} for 
$4\times 2$ codes followed by simulation results in Section \ref{sec:simu}. In Section \ref{sec:codes6} the results are extended for more transmit antennas and explicit constructions are provided for $6\times 3$ and 
$6\times 2$ codes. 

Further generalizations are  provided in Section \ref{sec:furthergen}, where it is also shown that the existence result can be made explicit via conjugations of the familiar left-regular representation. Section \ref{sec:conc} concludes the paper. In Appendix, relevant algebraic results related to central simple algebras and Hasse invariants are presented. 

The main contributions of this paper are listed below.
\begin{itemize}

\item General methods to produce space-time lattice codes with the NVD property and given geometric structure are given.

\item A unified construction of families of CDAs that can be embedded into  matrix rings of the Hamiltonian quaternions $M_k(\Hv)$ is provided. The underlying algebraic principles are studied in full detail. It is then demonstrated how such a structure can be beneficial in the decoding. The generality of the constructions is in contrast to the present \emph{ad hoc} constructions available in the literature.

\item A complete solution to the discriminant minimization problem \cite{HLRV2} for division algebras with arbitrary centers is given. As an application a normalized minimum determinant bound for code lattices in $M_k(\Hv)$ is derived from the algebraic  results.

\item We mainly consider the $4\times 2$ MIDO case, but also provide constructions for the $6\times 2$ and $6\times 3$ cases. The methods are generalizable to any even number of Tx antennas.
\item The main difference with other fast-decodable MIDO codes is that all the proposed codes have the NVD property. The proofs for the NVD are based on the underlying algebraic structure of the code and hold for infinite constellations. This can be seen as an improvement for \cite{SR-fast}, where the NVD is conjectured by computing the minimum determinant for certain finite QAM alphabets.
\item We build explicit codes that have 25-37.5\% reduced decoding complexity for general constellations, and whose performance is comparable to the best known MIDO codes. Such complexity is among the best known for the MIDO channel, and can be further reduced by using a  symmetric alphabet --  a square QAM alphabet, for instance. No fast-decodable MIDO codes with provable NVD other than the ones in this paper have been reported.

\end{itemize}

\subsection{Notations}
Throughout the paper, we will use the following notations:
\begin{itemize}
\item Tx for transmit antennas, Rx for receive antennas,
\item $n_t\times n_r$ for a channel with $n_t$ Tx and $n_r$ Rx antennas,
\item $(n\times k)$ for matrix dimensions,  
\item boldface lowercase letters for vectors, \emph{e.g.} $\mathbf{g}=(g_1,\ldots,g_t)$ or $\mathbf{g}=(g_1,\ldots,g_t)^T$,
\item capital letters for matrices, \emph{e.g.} $X$ or $M$,
\item $x^*$ for the complex conjugate of $x$, $X^*$ for element-wise conjugation in a matrix $X$, and $X^\dag$ for the Hermitian conjugate of $X$,
\item calligraphic letters for algebras, \emph{e.g.} $\mathcal{A}$, 
\item $E/K$ for number field extensions and $\sigma$ for the generator of a cyclic Galois group $\Gal(E/K)$. Note that $K$ is also used for the rank of a lattice in some instances, but this should cause no danger of confusion. 
\item The field norm from $E$ to $K$  is denoted by $$\Nc_{E/K}(x)=x\sigma(x)\cdots\sigma^{n-1}(x)\in K,$$ where $n=\#\Gal(E/K)$. 
\end{itemize}
%*********************************************************************%
%
% LATTICES
%
%*********************************************************************%

\section{Space-time lattice codes}
\label{sec:STC}

We start with as general a definition of a space-time code as possible,
and motivate why we focus our attention to {\em space-time lattice} codes,
which furthermore can be decoded via sphere decoder, a universal decoder
for lattice codes. We explain in detail how this is done.

%*********************************************************************%
\subsection{Definitions}
\label{subsec:def}

Abstractly, a space-time codeword $X$ is an $(n\times k)$ matrix
with coefficients in $\CC$, where $n$ corresponds to the number of
transmit antennas, and $k$ is the coherence time (or delay) during which the
channel is assumed constant. 
We will, in this paper, concentrate on the case $k=n$, so that a
space-time code is a square matrix, corresponding to minimum delay
codes.
\begin{definition}\label{def:stbc}
A {\em space-time code} $\Cc$ is a set of $(n\times n)$ complex matrices.
We often use the abbreviation {\em STBC} for {\em space-time block code}.
\end{definition}

The space $M_n(\C)$ of $(n\times n)$ matrices with complex coefficients
is a vector space of dimension
\[
\dim_\RR(M_n(\CC))=2n^2
\]
over the reals. Therefore, for every code $\Cc \subseteq M_n(\C)$, we can
consider, following \cite{RaRa}, the subspace $\langle \Cc \rangle$ spanned
by the matrices of $\Cc$. It has an  $\R$-basis consisting of $K$ matrices,
$1\leq K \leq 2n^2$, so that each matrix $X$ in $\Cc$ can be uniquely written
as
\begin{equation}\label{eq:ld}
X=\sum_{i=1}^{K}g_i B_i,
\end{equation}
where $B_i$ are some basis matrices and $g_i$ are real numbers.
Once the basis matrices $\{B_1,\ldots,B_K\}$ are given, a space-time code $\Cc$
is defined by the values that $g_i$, $i=1,\ldots,K$, can take. We write
\[
\gv=(g_1,\ldots,g_K)
\]
and let $\gv$ take its values in $\Gc\subseteq \R^K$, so that
\begin{equation}\label{eq:CG}
\Cc=\{\sum_{i=1}^{K}g_i B_i\,|\, \gv=(g_1,\dots,g_K) \in \Gc \,\,\}.
\end{equation}
Typically, $\Gc$ corresponds to a choice of constellation points.
For example, if a size $Q$ pulse amplitude modulation ($Q$-PAM) is used, then
$\Gc$ is the Cartesian product of $K$ times
$$\{-Q+1,\ldots,-3,-1,1,3,\ldots, Q-1\},$$
where  $Q\geq 2, 2|Q$.
The formulation in (\ref{eq:CG}) is not without recalling the notion of \emph{linear
dispersion codes} \cite{HH}, where codewords $X$ are similarly described by a
family of dispersion matrices $\{A_1,\ldots, A_K\}$: $X=\sum_{i=1}^Kg_iA_i$,
for some coefficients $g_i$ belonging to a symmetric set. The critical difference is in $\{ B_1,\ldots,B_K \}$
being linearly independent, and thus really forming an $\RR$-basis for $\langle\Cc\rangle $.
It consequently makes sense to speak of dimension of $\langle\Cc\rangle $, which yields
the following definition of rate \cite{RaRa}:
\begin{definition}\label{def:R1}
The {\em dimension rate} $R_1$ of the code $\Cc$ is given by
$$
R_1=\frac{\dim_\RR(\langle \Cc \rangle)}{n}=\frac{K}{n}
$$
(real) dimensions per channel use.
\end{definition}

Since $1\leq K \leq 2n^2$, we immediately see that the maximum rate achievable
for square matrices is $2n$.
One should note that this is not the common definition of a \emph{code rate} (also used in this paper until now), which
usually counts how many complex symbols (\emph{e.g.} QAM symbols) are transmitted in a codeword. With our notation, the common code rate would be $R_1/2\leq n$.

The data rate in bits per channel use (bpcu) is defined as follows.
\begin{definition}
The {\em bit rate} $R_2$ of the code $\Cc$ is
$$
R_2=\frac{\log_2(|\Cc|)}{n}
$$
bpcu.
\end{definition}

While the above considerations have been done in full generality, several
years of research on space-time coding have shown that good space-time codes
enjoy special properties. Following \cite{TSC}, getting fully diverse codes
has become the first code design criterion. That is, we require
\begin{equation}\label{eq:fulldiv}
\det(X-X') \neq 0,~X\neq X'\in\Cc.
\end{equation}
From \cite{SRS} it is known that the best way to actually deal with
this constraint is to first assume that the space-time code considered forms
an additive group, so that
\begin{equation}\label{eq:add}
X\pm X' \in \Cc,
\end{equation}
which simplifies (\ref{eq:fulldiv}) to
\[
\det(X)\neq 0,~X\neq {\bf 0},
\]
a much more tractable constraint. We note that $\Cc$ as defined in
(\ref{eq:CG}) is not necessarily linear, but of course
$\langle \Cc \rangle$ is.
From the linearity imposed on $\Cc$ by (\ref{eq:add}), we are only one
step away from having a {\em space-time lattice code}. Recall that
\begin{proposition}
An infinite discrete group of matrices in $M_n(\C)$ is a lattice.
\end{proposition}

We can thus safely assume that infinite space-time codes have a lattice structure,
since the discreteness condition can be translated by asking the Euclidean
distance between each pair of codewords to be greater than $r$, for a fixed
non-zero $r$. This formalizes the natural assumption that codewords should not
be chosen too close to each other.

\begin{definition}
A {\em space-time lattice code} $\Cc \subseteq M_n(\C)$ has the form
$$
\Z B_1\oplus \Z B_2 \cdots \oplus \Z B_K,
$$
where the matrices $B_1,\dots, B_K$ are linearly independent, \emph{i.e.}, form a lattice basis, and $K$ is
called the \emph{rank} of the lattice. We may also call $K$ the \emph{dimension} of the code, but do not confuse this with the dimension of the lattice.
\end{definition}

For the actual transmission, a finite subset of codewords from $\Cc$ is picked 
by restricting the integer coefficients to some set $\Gc$, as in (\ref{eq:CG}). From now on, we will consider only space-time lattice codes and may call
them space-time codes for short.

As recalled above, full diversity is the first design criterion for space-time
codes. Once achieved, meaning for lattice codes that
\[
\det(X)\neq 0,~X\neq {\bf 0},
\]
the next criterion is to maximize the minimum determinant of the code.
\begin{definition}\label{mindet}
The {\em minimum determinant} $\mindet{\Cc}$ of a space-time code
$\Cc \subset M_n(\C)$ is defined to be
\[
\mindet{\Cc}=\inf_{X\neq{\bf 0}}|\det(X)|,~ X\in\Cc.
\]
\end{definition}
\begin{definition}\cite{BR}\label{def:NVD}
If the minimum determinant of the lattice is non-zero, we say that the code
has a \emph{non-vanishing determinant} (NVD) .
\end{definition}

The NVD property means that, prior to SNR normalization, the lower bound
on the minimum determinant does not depend on the size of the constellation used.

%*********************************************************************%
\subsection{Sphere decoding}

Let $X$ be a space-time lattice codeword.
We can flatten $X\in M_n(\C)$ to obtain a $2n^2$-dimensional real
vector $\xv$ by first forming a vector of length $n^2$ out of the entries
(\emph{e.g.} row by row, or vectorizing that is column by column) and then replacing
each complex entry with the pair formed by its real and imaginary parts.
This defines a mapping $\alpha$ from $M_n(\C)$ to $\R^{2n^2}$:
\begin{equation}\label{eq:alpha}
\alpha: X \mapsto \xv=\alpha(X)
\end{equation}
which is clearly $\R$-linear:
\begin{equation}\label{eq:Rlin}
\alpha(rX+r'X')=r\alpha(X)+r'\alpha(X'),~r,r'\in\R.
\end{equation}
Let $||X||_F = \sqrt{\tr(X^\dagger X)}$ denote the Frobenius norm of $X$.
Note that the following equality holds:
\begin{equation}\label{eq:alphaiso}
||X||_F = \sqrt{\sum_{i=1}^n\sum_{j=1}^n|x_{ij}|^2}= ||\alpha(X)||_E,
\end{equation}
where $||\cdot||_E$ denotes the Euclidean norm of a vector.
This makes $\alpha$ an isometry.

The space-time code $X\in M_n(\C)$ is transmitted over a
coherent Rayleigh fading channel with perfect channel state information at
the receiver (CSIR):
\[
Y=HX+V,
\]
where $H$ is the channel matrix and $V$ is the Gaussian noise at the receiver.
Maximum-likelihood (ML) decoding consists of finding the codeword $X$ that
achieves the minimum of the squared Frobenius norm
\begin{equation}
\label{frob-min}
d(X)=||Y-HX||_F^2.
\end{equation}
This search can be performed using a real sphere decoder (see \emph{e.g.} \cite{VB}). Since
this paper focuses on MIDO codes and for the sake of simplicity, we will
now exemplify the computation of a $(4\times 4)$ MIDO code matrix $X$, that is, we
 consider 4 Tx antennas and 2 Rx antennas and the channel
\begin{equation}\label{eq:channel}
Y_{2\times 4}=H_{2\times 4}X_{4\times 4}+V_{2\times 4}.
\end{equation}
A $(4\times 4)$ MIDO code can transmit up to 8 complex (say QAM) information
symbols, or equivalently 16 real (say PAM) information symbols. Following (\ref{eq:CG}),
the encoding can thus be written as mapping the PAM vector
\[
\gv=(g_1,\dots,g_{16})^T
\]
into a $(4\times4)$ matrix
$$
X=\sum_{i=1}^{16} g_i B_i,
$$
where the basis matrices $B_i$, $i=1,\ldots,16,$ define the code. Let us emphasize
again that by basis matrices, we really mean a $\ZZ$-basis of the code seen as a lattice.
From (\ref{eq:channel}), the received matrix $Y$ can be expressed as
$$
Y_{2\times 4}=H(\sum_{i=1}^{16}g_{i}B_i) +V=\sum_{i=1}^{16}g_i(HB_i)+V.
$$
In order to  perform real sphere decoding, we have to transform this complex
channel equation into a real one, which can be done via the mapping $\alpha$
defined in (\ref{eq:alpha}).
The matrix $Y_{2\times 4}=(y_{i,j})$ can be turned into a real valued vector
$\yv$ in $\RR^{16}$ by the transformation
\[
\alpha(Y)=\yv=[\yv_1,\yv_2]^T
\]
with
\[
\begin{array}{l}
\yv_1=(\Re(y_{1,1}), \Im (y_{1,1}),\ldots, \Re(y_{1,4}),\Im(y_{1,4}))\\
\yv_2=(\Re(y_{2,1}), \Im (y_{2,1}),\ldots, \Re(y_{2,4}),\Im(y_{2,4})).
\end{array}
\]
The matrices $HB_i\in M_{4\times2}(\CC)$ are then similarly turned into
vectors $\bv_i\in \RR^{16}$:
\[
\alpha(HB_i)=\bv_i,~i=1,\ldots,16,
\]
so that $d(X)$ can be expressed as
\[
\begin{array}{lcll}
d(X) &=& ||Y-HX||_F^2  & \mbox{ by }(\ref{frob-min}) \\
     &=& ||\alpha(Y-HX)||_E^2 & \mbox{ by }(\ref{eq:alphaiso})\\
     &=& ||\alpha(Y)-\alpha(HX)||_E^2 & \mbox{ by }(\ref{eq:Rlin})\\
     &=& || \yv-\sum_{i=1}^{16}g_i\bv_i||_E^2.
\end{array}
\]
From this we finally get
\begin{equation}\label{eq:dlatt}
d(X) = || \yv -B\gv ||_E^2,
\end{equation}
where
\[
B=(\bv_1, \bv_2,\dots, \bv_{16}) \in M_{16\times16}(\RR).
\]
This shows that the decoding of a space-time lattice code $\Cc$ with a basis
$\{B_1,\ldots,B_K \}$ is equivalent
to the decoding of a 16-dimensional real lattice $\Lambda(\Cc)$ described by
the generator matrix $B$: $\Lambda(\Cc)=\{ \xv = B\gv ~|~\gv \in \ZZ^n\} $.

%*********************************************************************%
%
% FAST-DECODABILITY
%
%*********************************************************************%

\section{Fast-decodable space-time codes}
\label{sec:FD}

We are now ready to explain the notion of fast decodability of space-time
lattice codes when using sphere decoding. We will then give a few examples
that will motivate the rest of the paper.

%*********************************************************************%
\subsection{Fast sphere decoding}\label{fastdec}

The first step of the sphere decoder is to perform a QR decomposition
of the lattice generator matrix $B$, $B=QR$, with $Q^\dagger Q=I$, to reduce
the computation of
\[
d(X)= || \yv -B\gv ||_E^2
\]
as in (\ref{eq:dlatt}) to
\begin{equation}\label{eq:dR}
d(X)=|| \yv-QR\gv||_E^2=||Q^\dagger\yv-R\gv||_E^2
\end{equation}
where $R$ is an upper right triangular matrix. The number and position of
non-zero elements in the upper right part of $R$ will determine the complexity
of the sphere decoding process \cite{BHV,SR-fast}.

The worst case is of course given when the matrix $R$ is a full upper right
triangular matrix. This motivates the following definition of worst case
sphere decoding complexity:
\begin{definition}\cite[Def.\,2]{BHV}
\label{Vit_comp}
Let $S$ denote the real alphabet in use, and let $\kappa$ be the number of
independent real information symbols from $S$ within one code matrix.
The {\em ML decoding complexity} is the minimum number of values of $d(X)$ in
\eqref{eq:dR} that should be computed while performing  ML decoding. This
number cannot exceed $|S|^\kappa$, the complexity of the exhaustive-search ML
decoder (or $|S|^{\kappa/2}$ for a complex alphabet $S$).
\end{definition}

\begin{definition} The exponent $\kappa$ (resp. $\kappa/2$) is referred to as the {\em dimension of a real (resp. complex) sphere decoder}. If the structure of the code is such that $\kappa$ decreases, we say that the
code is {\em fast-decodable}. In this paper, we always refer to the dimension of a real sphere decoder.
\end{definition}

In the MIDO case (\ref{eq:channel}), where $S$ is a real PAM alphabet (and hence $|S|$ is the number of PAM symbols in use), the worst case complexity is $|S|^{16}$.
A typical improvement in $\kappa$ can be obtained if the left upper corner of
the matrix
\[
R =
\left(
\begin{array}{cc}
R^{1,1} & R^{1,2}\\
R^{2,1} & R^{2,2}
\end{array}
\right)
\]
from the QR decomposition of $B$ has the form
\begin{equation}\label{eq:upperleft}
R^{1,1}=
\left(
\begin{array}{cccccccc}
\star&\star&\star&\star&0&0&0&0\\
0&\star&\star&\star&0&0&0&0\\
0&0&\star&\star&0&0&0&0\\
0&0&0&\star&0&0&0&0\\
0&0&0&0&\star&\star&\star&\star\\
0&0&0&0&0&\star&\star&\star\\
0&0&0&0&0&0&\star&\star\\
0&0&0&0&0&0&0&\star\\
\end{array}\right),
\end{equation}
where $\star$ denotes any non-zero element. Indeed, in this case:
\begin{enumerate}
\item
We start the sphere decoding by going through every combination of the 8
last real symbols $g_9,\ldots,g_{16}$ (we are not choosing the ones that give
the minimal metric yet, we go through all the options since we do not know
how the last 8 symbols will affect the total minimization problem). This
corresponds to treating the matrix $R^{2,2}$, and has cost $|S|^8$.
\item
We then look at the first 8 symbols $g_1,\ldots,g_8$, corresponding to
the matrix $R^{1,1}$, and for every possible choice of 8-tuples,
$(g_9,\ldots,g_{16} )$, we decode separately $g_1,\ldots,g_4$ and
$g_5,\ldots,g_8$ thanks to the structure of $R^{1,1}$, which has complexity
$2|S|^4$.
\end{enumerate}
Altogether, the above structure allows to decode the PAM symbols
$g_1,g_2,g_3,g_4$ independently of the symbols $g_9,g_{10},g_{11},g_{12}$,
yielding a worst case complexity of $|S|^{12}$ (or more precisely $2|S|^{12}$)
for the real sphere decoding process instead of the full complexity order
of $|S|^{16}$.

The natural question to ask is thus the design of codes (that is, of the basis
matrices $B_i$) that yield a sparse matrix $R$.
To address this question, we further study the structure of the matrix $R$.
By definition of the QR decomposition of the matrix
$B=(\bv_1,\ldots,\bv_{16})$, we have that
\[
R=
\left(
\begin{array}{cccc}
\langle\ev_1,\bv_1\rangle & \langle\ev_1,\bv_2 \rangle &\ldots&\langle\ev_1,\bv_{16} \rangle\\
0                         & \langle\ev_2,\bv_2 \rangle &\ldots&\langle\ev_2,\bv_{16} \rangle\\
0                         & 0                          &      &\langle\ev_3,\bv_{16}\rangle\\
0                         & 0                          &\ddots & \vdots                     \\
0                         & 0                          &  & \langle \ev_{16},\bv_{16} \rangle
\end{array}
\right)
\]
where
\begin{eqnarray*}
\ev_1 &=&\frac{\bv_1}{||\bv_1||}\\
\ev_2 &=&\frac{\bv_2-\proj_{\ev_1}\bv_2}{||\bv_2-\proj_{\ev_1}\bv_2||}\\
\vdots &&\\
\ev_k &=&\frac{\bv_k-\sum_{j=1}^{k-1}\proj_{\ev_j}\bv_j}{||\bv_k-\sum_{j=1}^{k-1}\proj_{\ev_j}\bv_j||}
\end{eqnarray*}
and
\[
\proj_{\ev}\bv=\frac{\langle \ev,\bv \rangle}{\langle \ev,\ev \rangle}\ev.
\]
The notation $\langle\cdot,\cdot\rangle$ stands for the usual inner product.
Thus having the upper left part of $R$ to look like (\ref{eq:upperleft})
means that
\[
\langle \bv_i, \bv_j\rangle=0,~1\leq i \leq 4,~5\leq j \leq 8,
\]
or equivalently, by recalling that $\bv_i=\alpha(HB_i)$
\[
0=\langle \alpha(HB_i),\alpha(HB_j) \rangle=\Re(\tr(HB_i(HB_j)^\dagger)).
\]
The second equality is true in general and can be shown by a direct computation:
\begin{equation}\label{eq:retrprod}
\langle \alpha(A),\alpha(B)\rangle=\Re(\tr( AB^{\dagger})).
\end{equation}

We have now connected the decoding complexity to the code design. The above
computations showed that if the 16 basis matrices $B_1,\ldots,B_{16}$ satisfy
\[
0=\Re(\tr(HB_i(HB_j)^\dagger)),~1\leq i \leq 4,~5\leq j \leq 8,
\]
the worst case sphere decoding complexity is of the order of $|S|^{12}$.
This suggests further improvement: the current process manages to separate
the information symbols into two groups, which could be repeated.  Assume that
we could further have
\[
0=\Re(\tr(HB_i(HB_j)^\dagger)),~1\leq i \leq 2,~3\leq j \leq 4
\]
and
\[
0=\Re(\tr(HB_i(HB_j)^\dagger)),~5\leq i \leq 6,~7\leq j \leq 8.
\]
\begin{enumerate}
\item[3)]
As earlier, we start the sphere decoding with the matrix $R^{2,2}$ and go
through all the possibilites for the 8 last real symbols $g_9,\ldots,g_{16}$,
for a cost of $|S|^8$.
\item[4)]
For the first 8 symbols $g_1,\ldots,g_8$ corresponding to
the matrix $R^{1,1}$, we first separate $g_1,\ldots,g_4$ and
$g_5,\ldots,g_8$, after which we decode independently $\{g_1,g_2\}$,
$\{g_3,g_4\}$, $\{g_5,g_6\}$ and $\{g_7,g_8\}$, each of these costing
$|S|^2$.
\end{enumerate}
The worst case complexity is then $4|S|^8|S|^2=4|S|^{10}$.

\begin{remark}
It is possible to further reduce the (ML) complexity by using the so-called \emph{hard-limiting}, see \cite[Section VI, p. 924 (1-2)]{SR-fast}. In this case, the complexity will be  $4|S|^{4.5}$, where $|S|$ is the size of a complex signal constellation. However, this is only possible when a square constellation (\emph{e.g.} $Q^2$-QAM) can be employed, \emph{i.e.}, the constellation is a cartesian product of two real constellations (\emph{e.g.} $Q$-PAM). 
\end{remark}

%*********************************************************************%
\subsection{Examples from the ring of Hamiltonian quaternions}
\label{subsec:hamil}

To illustrate the material explained above, let us start with the
Alamouti code \cite{Alam}, \emph{i.e.}, codewords of the form
\[
X=\left(
\begin{array}{cc}
x_1 & -x_2^* \\
x_2 & x_1^*
\end{array}
\right)=
\left(
\begin{array}{cc}
g_1+ig_2 & -g_3+ig_4 \\
g_3+ig_4 & g_1-ig_2
\end{array}
\right),
\]
where $x_1,x_2$ are QAM symbols and $\gv=(g_1,g_2,g_3,g_4)$ is the PAM symbol 
vector.
A decomposition into basis matrices $B_1,B_2,B_3,B_4$ is given by
\[
X=g_1B_1+g_2B_2+g_3B_3+g_4B_4,
\]
where
\[
B_1=\left(
\begin{array}{cc}
1 & 0 \\
0 & 1
\end{array}
\right),~
B_2=\left(
\begin{array}{cc}
i & 0 \\
0 & -i
\end{array}
\right),
\]
\[
B_3=\left(
\begin{array}{cc}
0 & -1 \\
1 & 0
\end{array}
\right),~
B_4=\left(
\begin{array}{cc}
0 & i \\
i & 0
\end{array}
\right).
\]
We assume transmission through a MISO channel described by the vector
\[
H=(h_1,h_2)
\]
so that $\alpha(HB_i)$, $i=1,2,3,4$, is given by
\begin{eqnarray*}
\bv_1=\alpha(HB_1)&=&(\Re(h_1),\Im(h_1),\Re(h_2),\Im(h_2))^T,\\
\bv_2=\alpha(HB_2)&=&(-\Im(h_1),\Re(h_1),\Im(h_2),-\Re(h_2))^T,\\
\bv_3=\alpha(HB_3)&=&(\Re(h_2),\Im(h_2),-\Re(h_1),-\Im(h_1))^T,\\
\bv_4=\alpha(HB_4)&=&(-\Im(h_2),\Re(h_2),-\Im(h_1),\Re(h_1))^T.
\end{eqnarray*}
We finally get
\[
B=\alpha(HX)=[\bv_1,\bv_2,\bv_3,\bv_4],
\]
and since $\langle\bv_i,\bv_j \rangle=0$ for $i\neq j$, the QR decomposition
of $B$ is of the form
\[
B=\left(\frac{1}{c}B\right)\left(cI_4\right)=QR,
\]
where
\[
c=\sqrt{\Re(h_1)^2+\Im(h_1)^2+\Re(h_2)^2+\Im(h_2)^2}
\]
is a normalization factor which makes $Q$ orthonormal. The matrix $R$ is indeed
upper right triangular, with in fact only zeroes above its diagonal.
Thus the worst case decoding complexity of such a code is the size of the QAM alphabet, that is, of linear order.

Finding basis matrices with similar properties as those of the Alamouti code seems a difficult task.
The question is in general to find families of matrices $\{B_1,\ldots,B_K\}$ which are
{\em orthogonal} in the sense that $\langle \alpha(B_i), \alpha(B_j) \rangle =0 $, $i\neq j$,
and will keep this property even after multiplication by an arbitrary channel matrix $H$.
Let us start modestly and wonder whether we could find such a pair of matrices $B,B'\in M_n(\C)$
whose orthogonality will resist a channel matrix $H \in M_{k\times n}(\C)$,
where  $n\geq k$. Using (\ref{eq:retrprod}), we need to check that
\[
0=\langle \alpha(HB),\alpha(HB') \rangle=\Re(\tr(HB(HB')^\dagger)).
\]
As a first example, take
$$
B=
\begin{pmatrix}
x_1&0\\
0&x_1^*
\end{pmatrix}
\mbox{ and }
B'=
\begin{pmatrix}
0&-x_2^*\\
x_2&0
\end{pmatrix},
$$
where $x_1,x_2\in\CC$.  These two matrices clearly satisfy the orthogonality relation $ \langle \alpha(B),\alpha(B') \rangle=0$.
Now pick an arbitrary complex matrix
$$
H=
\begin{pmatrix}
h_1&h_2\\
h_3&h_4
\end{pmatrix}.
$$
A direct calculation shows that
\begin{eqnarray*}
&&\tr(HB(HB')^{\dagger})\\
&=&
x_1 h_1 x_2^* h_2^*
-h_2 x_2 h_1^*x_1^*
+x_1 h_3 x_2^* h_4^*
-x_2 h_4 h_3^*x_1^*\\
&=&
i\Im (x_1 h_1 x_2^* h_2^*)+i\Im(x_1 h_3 x_2^* h_4^*)
\end{eqnarray*}
so that
\[
\Re( \tr(HB(HB')^{\dagger}))=0,
\]
independently of the matrix $H$.

As a second example, consider
$$B=
\begin{pmatrix}
x_1&  0&   0&  0\\
0&  x_1*& 0&  0\\
0&  0&    x_3&  0\\
0 & 0&   0&  x_3*
\end{pmatrix},
$$
$$B'=
\begin{pmatrix}
0 &-x_2*&  0&  0\\
x_2 &  0&   0&  0\\
0 &  0&   0& -x_4*\\
0 &  0  & x_4&  0
\end{pmatrix}
$$
and
$$H=
\begin{pmatrix}
h_1 &h_2&  h_3&  h_4\\
h_5 &  h_6&   h_7&  h_8\\
\end{pmatrix}.
$$
We can similarly see that
$\Re(\tr(HB(HB')^{\dagger}))=0$.

The notable thing however is that both examples are closely
related to the Alamouti code (the first example being really included in it).
This is not a surprise, since most of the work available on fast ML decodability
tries to actually exploit the code structure. To pursue our
investigation on fast decodability, we now need to focus on algebraic
constructions of space-time lattice codes from division algebras.

%**********************************************************************%
%
% DIVISION ALGEBRA STBC
%
%**********************************************************************%

\section{Space-time codes from division algebras}
\label{sec:cda}

%**********************************************************************%
\subsection{Background}
\label{background}
Since the work of Sethuraman et al. \cite{SRS}, a standard algebraic technique
to build space-time block codes is to use cyclic division algebras over number
fields (that is, finite extensions of the field $\Q$).
For the sake of completeness, we will start by recalling the formal definition
of a cyclic algebra, after which we will provide an illustrative example, rather
than redo the whole theory, which the reader can find in \cite{SRS}, or in the
tutorial \cite{OBV}.

\begin{definition}\label{cyclic}
Let $K$ be an algebraic number field and assume that $E/K$ is a cyclic Galois
extension of degree $n$ with Galois group
$\Gal(E/K)=\left\langle \sigma\right\rangle$. We can now define an associative $K$-algebra
$$
\mathcal{A}=(E/K,\sigma,\gamma)=E\oplus uE\oplus u^2E\oplus\cdots\oplus u^{n-1}E,
$$
where $u\in\mathcal{A}$ is an auxiliary generating element subject to the relations
$xu=u\sigma(x)$ for all $x\in E$ and $u^n=\gamma\in K^*$, where $K^*$ denotes
$K$ without the zero element. 

The element $\gamma$ is often called a \emph{non-norm} element due to its relation to the invertibility of the elements of $\mathcal{A}$. Namely, if there exists no element $x\in E$ such that its norm 
would be $\Nc_{E/K}(x)=\gamma^t$, where $t\in\Z_+$ is a proper divisor of $n$,
then $\mathcal{A}$ will be a division algebra \cite[Prop. 2.4.5]{Camithesis}. This result is a straightforward simplification of a theorem by Albert \cite{AA}.
\end{definition}

Space-time codewords are obtained by considering matrices of left multiplication by
an element of $\Ac$ in the above basis.

Let us see how the coding is done more concretely through an example. We first need
a number field  $E$ of degree $n$ whose Galois group is cyclic. For example, take
$\zeta_5=e^{2i\pi /5}$ a primitive 5th root of unity, and consider the number
field $E=\QQ(i,\zeta_5)$ over $K=\QQ(i)$, given by
\[
\QQ(i,\zeta_5)=\{x=a+b\zeta_5+c\zeta_5^2+d\zeta_5^3,~a,b,c,d \in \QQ(i)\}.
\]
It is of degree 4 (\emph{i.e.}, of dimension 4 as a vector space) over $\QQ(i)$.
Let us assume that we want to encode QAM symbols.
Since they can be seen as elements in $\ZZ[i]\subset \QQ(i)$, we
have that one element $x$ in $\QQ(i,\zeta_5)$ encodes 4 QAM symbols,
namely $a,b,c,d$, as linear combinations in the given basis.
The Galois group of $\QQ(i,\zeta_5)/\QQ(i)$ describes
maps that permute $\zeta_5$ and its conjugates $\zeta_5^j$, $j=2,3,4$ while
fixing $\QQ(i)$. If $\sigma(\zeta_5)=\zeta_5^2$, we have that
\[
\sigma^2(\zeta_5)=\zeta_5^4,~\sigma^3(\zeta_5)=\zeta_5^3,
~\sigma^4(\zeta_5)=\zeta_5
\]
yielding a cyclic Galois group. We now build an associative algebra $\Ac$
based on $E$. As a vector space, $\Ac$ can be seen as a sum of $n$ copies of
the chosen number field $E$ of degree $n$. In our example, this gives
\[
\Ac=\QQ(i,\zeta_5)\oplus u\QQ(i,\zeta_5)\oplus u^2\QQ(i,\zeta_5)
\oplus u^3\QQ(i,\zeta_5)
\]
where $\{1,u,u^2,u^3\}$ forms a basis and $\gamma=u^4$ must be an element of the base
field $\QQ(i)$, say $u^4=i$. A space-time block code can be obtained by
considering the matrix of left multiplication in this given basis. If
$x=x_0+ux_1+u^2x_2+u^3x_3 \in \Ac,~x_0,x_1,x_2,x_3\in\QQ(i,\zeta_5)$, then
its corresponding multiplication matrix is
\begin{equation}\label{eq:matrix}
X=
\left(
\begin{array}{rrrr}
x_0 & i\sigma(x_3)& i\sigma^2(x_2) & i\sigma^3(x_1)\\
x_1 & \sigma(x_0) & i\sigma^2(x_3) & i\sigma^3(x_2)\\
x_2 & \sigma(x_1) & \sigma^2(x_0)  & i\sigma^3(x_3)\\
x_3 & \sigma(x_2) & \sigma^2(x_1)  & \sigma^3(x_0)\\
\end{array}
\right)
\end{equation}
where the factor $i$ comes from $u^4=i$ and $\sigma^j$,
$j=1,2,3,4$, are the elements of the Galois group, appearing due to the
non-commutative multiplication defined on $\Ac$ by $xu=u\sigma(x)$ for
$x\in E$.

Let $\Cc$ be the codebook formed by codewords $X$ of the above form.
For it to be fully diverse, recall from (\ref{eq:fulldiv}) that it is enough to have
\[
\det(X'-X'')\neq 0
\]
for $X'\neq X''$ in $\Cc$, or equivalently, by linearity since we are
considering space-time lattice codes
\[
\det(X)\neq 0
\]
for $X\neq {\bf 0}$ in $\Cc$. This can be obtained by asking for $\Ac$ to be
a division algebra, property that depends on the choice of the value of
$\gamma$ (or $\gamma=i$ in our example).
If there exists no element $a\in \QQ(i,\zeta_5)$ such that its norm
is $i$ or $i^2$, \emph{i.e.}, $\Nc_{\QQ(i,\zeta_5)/\QQ(i)}(a)= i,$ or  $-1$,
then $\Ac$ will be a division algebra \cite{AA,Camithesis}.

Let us check that $\Ac$ is indeed a division algebra. Note for this purpose
that $\QQ(\zeta_5+\zeta_5^{-1})=\QQ(\sqrt{5})$ is a subfield of $\QQ(\zeta_5)$.
Suppose now that there exists an element $a\in E$ such that
$\Nc_{\QQ(i,\zeta_5)/\QQ(i)}(a)=i$, then, by transitivity of the norm
\[
\Nc_{\QQ(i,\zeta_5)/\QQ(i)}(a)
=\Nc_{\QQ(i,\sqrt{5})/\QQ(i)}\Nc_{\QQ(i,\zeta_5)/\QQ(i,\sqrt{5})}(a)=i,
\]
which implies the existence of an element
$b=\Nc_{\QQ(i,\zeta_5)/\QQ(i,\sqrt{5})}(a)$ such that
\[
\Nc_{\QQ(i,\sqrt{5})/\QQ(i)}(b)=i,
\]
a contradiction \cite{BRV}. 

%The case of a norm of $-i$ can be dealt similarly,
%since using the same argument of transitivity of the norm, it is enough to show that
%there cannot be an element with norm $-i$ over $\QQ(i,\sqrt{5})/\QQ(i)$. If there
%were an element $a$ with $N_{\QQ(\sqrt{5},i)/\QQ(i)}(a)=-i$, then $ia$ would have norm
%\[
%i^2N_{\QQ(\sqrt{5},i)/\QQ(i)}(a)=i,
%\]
%a contradiction.

The case of a norm of $-1$ is tougher though. However, there are several ways to deal with it.
We refer the reader to \cite[Section 8]{HLRV2}, where the proof used for the algebra $D_4$
can be used here verbatim.

We have thus constructed in our example a fully-diverse $(4\times 4)$
space-time code matrix. It furthermore has the non-vanishing
determinant property (see Definition \ref{def:NVD}), since the information symbols
are restricted to algebraic integers in $L$, and hence the minimum determinant
belongs to $\ZZ[i]$, yielding $\min_{X\neq {\bf 0}} |\det(X)|=1$ (cf. \cite{HLRV2}).

We conclude with two important invariants of central simple algebras. Central simple $K$-algebras
are algebras whose center is $K$ and which have only trivial two-sided ideals. Cyclic algebras are
particular cases of central simple algebras. We could have stated these definitions only for cyclic
algebras, but for the rest of this work, we will need them in more generality.
\begin{definition}
Let $\Ac$ be a central simple $K$-algebra. The {\em degree} of $\Ac$ is the integer
$\deg(\Ac)=\sqrt{\dim_K(\Ac)}$.
\end{definition}

\emph{Wedderburn's theorem} is a major theorem in the theory of central simple algebras,
which tells that every central simple algebra (and thus in particular every cyclic algebra)
is isomorphic to a matrix algebra over a central division $K$-algebra $\Dc$.

\begin{definition}
The {\em index} of $\Ac$ is the integer $\ind(\Ac)=\deg(\Dc)$ where $\Dc$ is the unique central
division $K$-algebra associated to $\Ac$ by Wedderburn's theorem.
\end{definition}

We have that $\ind(\Ac )~|~\deg(\Ac)$ and equality holds if and only if $\Ac$
is a division algebra.

%**********************************************************************%
\subsection{Examples}

Let us now consider a few well known examples of division algebra codes, and see how
they behave with respect to fast decodability.

The Alamouti code \cite{Alam} can be seen from an algebraic perspective as a cyclic division algebra
\begin{equation}\label{eq:Dalam}
\D_{Alam}=(\Q(i)/\Q,\sigma, -1),
\end{equation}
where $\sigma$ is the complex conjugation. This is a $\Q$-central division algebra of index $2$,
whose cyclic representation indeed yields codewords of the type
$$
\begin{pmatrix}
x_1&-x_2^*\\
x_2& x_1^*
\end{pmatrix},
$$
where $x_i$ are in $\Z[i]$ (that is, they are QAM symbols).

This algebra is more commonly known as the Hamiltonian quaternions
\[
\Hv =\{a+ib+jc+ijd\ |~a,b,c,d\in\RR\},
\]
\[
\mbox{ where }i^2=j^2=-1,~ij=-ji.
\]

Probably the most important property of this code is that, when used
over a MISO channel, its
worst case decoding complexity is linear, as was shown in Subsection \ref{subsec:hamil}.

Let us now consider the division algebra
\begin{equation}\label{eq:Dort}
\D_{ort}=(\Q(i, \sqrt{2})/\Q(\sqrt{2}), \sigma,-1)
\end{equation}
from \cite{HLL}. This is an index $2$ algebra with center $\Q(\sqrt{2})$.
It can be turned into a space-time code by mapping the element
$x=a_1+a_2\zeta_8 +u a_3 + u\zeta_8 a_4 \in \D_{ort}$ to a codeword
$X$ given by
\[
\left(
\begin{array}{cccc}
a_1+a_2 \zeta_8& -a_3^*-a_4^*\zeta_8^*& 0&0\\
a_3+a_4\zeta_8 & a_1^*+a_2^*\zeta_8^*  & 0&0\\
0              & 0& a_1-a_2\zeta_8 &-a_3^*+a_4^*\zeta_8^*\\
0              & 0& a_3-a_4\zeta_8 &a_1^*-a_2^*\zeta_8^*
\end{array}
\right),
\]
where  $a_j=g_{2j-1}+a_{2j}\in \Z[i]$, $j=1,2,3,4$. We can now write this in the form
\[
X=\sum_{j=1}^8g_jB_j,
\]
where $\gv=(g_1,\ldots,g_8)$ is the PAM symbol vector, and the basis matrices are
\[
B_1=\diag(1,1,1,1),~B_3=\diag(\zeta_8,\zeta_8^*,-\zeta_8,-\zeta_8^*),
\]
\[
B_2=\diag(i,-i,i,-i),~B_4=\diag(i\zeta_8,-i\zeta_8^*,-i\zeta_8,i\zeta_8^*),
\]
\[
B_5=\!
\begin{pmatrix}
0&-1 &   &   \\
1 &0 &   & \\
   &   &  0 &-1\\
   &   &  1 & 0
\end{pmatrix},
B_7=\!
\begin{pmatrix}
 0 &-\zeta_8^*    &   & \\
\zeta_8 & 0 &   & \\
  &   &  0 & \zeta_8^* \\
  &   &  -\zeta_8 & 0
\end{pmatrix},
\]
\[
B_6=\!
\begin{pmatrix}
 0 &i    &   & \\
  i & 0&   & \\
  &   & 0 &i  \\
  &   &  i & 0
\end{pmatrix},
B_8=\!
\begin{pmatrix}
0 & i\zeta_8^*   &   & \\
  i\zeta_8 & 0 &   & \\
  &   & 0 & -i\zeta_8^* \\
  &   &  -i\zeta_8 & 0
\end{pmatrix}.
\]
The decoding complexity of this code for a MISO channel is $2|S|^4$ instead of
the maximal complexity $|S|^8$. Indeed, write the channel $H=(h_1,h_2,h_3,h_4)$
as $(H_1,H_2)$ with $H_1=(h_1,h_2)$ and $H_2=(h_3,h_4)$, so that
\[
HB_i=
(H_1,H_2)
\left(
\begin{array}{cc}
B_i^{1,1} & {\bf 0}\\
{\bf 0} & B_i^{2,2}
\end{array}
\right)
=(H_1B_i^{1,1},H_2B_i^{2,2}),
\]
whence $\Re(\tr(HB_i(HB_j)^\dagger))$ simplifies to
\[
\Re(\tr(H_1B_i^{1,1}(B_j^{1,1})^\dagger H_1^\dagger)
+\tr(H_2B_i^{2,2}(B_j^{2,2})^\dagger H_2^\dagger)).
\]
The basis matrices are closely related to those of the Alamouti code
given in Subsection \ref{subsec:hamil}, and it is easy,
using the known orthogonality relations of the Alamouti basis matrices,
to see that
\[
\Re(\tr(HB_i(HB_j)^\dagger))=0,~i=1,2,3,4,~j=5,6,7,8,
\]
yielding an upper triangular matrix $R$ of the same form as in 
(\ref{eq:upperleft}), and consequently a decoding complexity of $2|S|^4$.

Our final example is the division algebra
$$
\A_2=(\Q(\sqrt{3})/\Q, \sigma, -1),
$$
where $\sigma(\sqrt{3})=-\sqrt{3}$. This algebra is of index 2 with center $\Q$, and  yields codewords of the form
$$
\begin{pmatrix}
x_1+x_2\sqrt{3}&-x_3 +x_4\sqrt{3}\\
x_3+x_4\sqrt{3}& x_1-x_2\sqrt{3}
\end{pmatrix},
$$
where $x_i \in \Z$.  However, as far as we know there is no existing method to
reduce the decoding complexity of this code.

\begin{table}%[thp!]
\begin{center} \begin{tabular}{|l|l|l|l|l|l|}
\hline
MISO code        & matrix        &  center      & index & $|S|^\kappa$ (real) & max $|S|^\kappa$ \\
\hline
$\Dc_{Alam}$&$(2\times2)$  & $\Q$           & 2   & $|S|$        & $|S|^4$      \\
$\Dc_{ort}$ &$(4\times4)$ & $\Q(\sqrt{2})$ & 2   & $|S|^4$      & $|S|^8$      \\
$\Ac_2$     &$(2\times2)$  & $\Q$           & 2   & $|S|^4$      & $|S|^4$\\
\hline
\end{tabular}
\caption{Code constructions: algebraic properties versus decoding complexity}
\label{taulukko3}
\end{center}
\end{table}

We already observed in Subsection \ref{subsec:hamil} that from the decoding perspective, it
might be beneficial for codes to inherit some of the special structure of the Alamouti code.
This study of different algebraic code structures seems to concur with the same conclusion,
expressed now in algebraic terms as: a code should be a subset of $M_k(\Hv)$ for some $k$.
However, which algebras exactly give fast decodability still seems unclear (see Table \ref{taulukko3}).
In the following section, we are going to answer  this question.

%******************************************************************************%
%
% EMBEDDINGS INTO H
%
%******************************************************************************%

\section{Embedding codes into matrix rings of the Hamiltonian quaternions}\label{embed}

We have so far discussed fast decodability of space-time codes via sphere decoding, and
through several heuristic examples concluded that codewords in rings $M_k(\Hv)$, for some $k$ and
$\Hv$ the Hamiltonian quaternions, are prone to offer orthogonality relations that induce fast sphere
decoding. Therefore our main interest is now to study space-time codes that are subsets of the
rings $M_k(\Hv)$. This will be characterized by the ramification of the cyclic algebra over which
the space-time code is built.

%*****************************************************************************%
\subsection{Embedding division algebras into $M_k(\Hv)$}\label{mainmapsection}

Let $K/\Q$ be an algebraic extension of degree $m$. We then have that
\[
m = r_1 + 2r_2,
\]
where $r_1$ is the number of real embeddings and $r_2$ the number of pairs of complex embeddings
of $K$. We call these embeddings the {\em infinite primes} of the field $K$ and the non-zero prime
ideals of the ring $\OO_K$ the {\em finite primes} of the field $K$. If the embedding is complex,
resp. real, we call it a {\em complex} resp. {\em real} prime.
To each prime $P$, finite or infinite, corresponds a local field $K_P$, obtained by completion
of $K$ with respect to the absolute value induced by $P$ (the same way $\RR$ is obtained from $\QQ$
by completion with respect to the usual absolute value).

Let $\A$ be a central division $K$-algebra of index and thus degree $n$. Consider
\[
\Ac_P = \Ac \otimes_K K_P
\]
a central simple $K_P$-algebra, which is known to be isomorphic to $M_r(\Dc)$ for some $r$ and some
central division $K_P$-algebra $\Dc$. We denote by $m_P$ the index of $\Ac_P$ and call it the
{\em local index} of $\Ac$ at $P$. We say that $P$ is ramified in $\Ac$ if $m_P>1$

Let us define the space $G(\C)_{n}\subseteq M_{n\times2n}(\CC)$ by
\[
G(\C)_{n}=\{ (B^*,B)\in M_{n\times2n}(\C)~|B\in M_n(\C)\}
\]
and $B^*=(b_{ij}^*)$.
Now $\Ac\otimes_\QQ\R$ is a semi-simple $\QQ$-algebra, and can thus be written as a Cartesian
product of simple subalgebras. Its center is $K\otimes_\QQ \RR$, which is isomorphic to copies
of $\RR$ or $\CC$: a copy of $\RR$ for each real embedding of $K$, and one of $\CC$ for each pair
of conjugate complex embeddings. The simple components of $\Ac\otimes_\QQ\R$ will thus have these
factors as centers, and will be either central simple algebras over $\RR$ or $\CC$: those over $\CC$
will be matrix algebras over $\CC$, while those over $\RR$ will be either matrix algebras over
$\RR$ if $\Ac$ is not ramified in the corresponding real prime, or matrix algebras over $\HH$ if $\Ac$
is ramified.
Formally, we obtain the isomorphism \cite{Bayer}
\begin{equation}\label{mainmap}
\Ac\otimes_\QQ\R \cong M_{n/2}(\Hv)^{\omega} \times M_{n}(\R)^{r_1 -\omega} \times G(\C)^{r_2},
\end{equation}
where $\omega$ is the number of real places where $\Ac$ ramifies. Therefore each element in $\Ac$ can
be seen as a concatenation of $\omega$ matrices in $M_n(\CC)$, $r_1-\omega$ matrices in $M_n(\RR)$ and
$r_2$ pairs of conjugate matrices in $M_n(\CC)$, or alternatively as a matrix in $M_{n\times nm}(\C)$,
recalling that $m=r_1+2r_2$.

The above isomorphism (\ref{mainmap}) implies an injection $\psi$
\begin{equation}\label{psikuvaus}
\Ac \hookrightarrow \diag(M_{n/2}(\Hv)^{\omega} \times M_{n}(\R)^{r_1 -\omega} \times G(\C)^{r_2}),
\end{equation}
where the diag-operator places the $i$th $(n\times n)$ block to the $i$th diagonal block of a matrix
in $M_{mn}(\C)$.
From (\ref{psikuvaus}), we now see that it is possible to embed a division algebra $\Ac$ into $M_k(\Hv)$ if and only if
\begin{equation}\label{eq:embed}
\psi:\Ac \hookrightarrow \mathrm{diag}(M_{n/2}(\Hv)^m),
\end{equation}
namely we must have $r_2=0$ and $r_1-\omega=0$. In summary, we have that
\begin{corollary}\label{cor:cond}
In order to be able to embed a division $K$-algebra $\Ac$ into $M_{n/2}(\Hv)$:
\begin{itemize}
\item
The center $K$ cannot have complex places, that is, it must be totally real ($r_1=m$).
\item
Combined with the equation $r_1-\omega=0$, we then have that $\omega=m$, so that all the infinite places
of $K$ must be ramified in $\Ac$.
\end{itemize}
\end{corollary}
Let us then suppose that $K$ is indeed a totally real number field.
We shall now give a simple family of cyclic $K$-algebras that fulfill the second condition above.

\begin{proposition}\label{CM}
Let $\A=(E/K, \sigma, \gamma)$ be a cyclic division algebra, where $E$ is a CM-field
(\emph{i.e.}, $E$ is a totally complex field containing a totally real field $E_1$ such that $[E : E_1] = 2$).
Let $\eta_1,\dots, \eta_m$ be the $\Q$-embeddings of  $K$. If $\eta_i(\gamma)$ is negative for any $\eta_i$,
then  all the infinite places of  $\A$ are ramified.
\end{proposition}

\begin{proof}
Let us suppose that $P_i$  is one of the infinite primes in the field $K$ and that $\eta_i$ is the corresponding $\Q$-embedding.
Let  $k$ be  the smallest possible positive power such that $\sigma^k$ fixes the totally real subfield  $E_1$ of $E$.
We then have \cite[Theorem 30.8]{R}
\begin{equation}\label{ramornot}
(E/K, \sigma, -\gamma)\otimes_{\QQ} K_{P_i} \sim (EK_{P_i}/K_{P_i},\sigma^k, -\eta_i(\gamma)),
\end{equation}
where $\sim$ refers to equivalence in the \emph{Brauer group}  $B(K_{P_i})$.
Because $P_i$ is a real prime, we can identify $K_{P_i}$ and $\R$, and similarly, $EK_{P_i}$ and $\C$,
so that from \eqref{ramornot}, we get  $\langle\sigma^k\rangle=\Gal(\C/\R)$.
Finally,
$$
(E/K, \sigma, -\gamma)\otimes_{\QQ} K_{P_i} \sim (\C/\R, \sigma^* , -\eta_i(\gamma)),
$$
where $\sigma^*$ is the complex conjugation and  $-\eta_i(\gamma)$ is a negative real number. The claim now follows
as $(\C/\R, \sigma^* , -\eta_i(\gamma))\cong \HH$.
\end{proof}

We point out that for rational numbers $r$ we have $\eta_i(r)=r$. Therefore a negative rational number is always
a suitable non-norm element if $\Ac$ is a division algebra.

\begin{exam}
The  algebras $\D_{ort}$ and $\D_{Alam}$ discussed above both fulfill  the conditions of Proposition \ref{CM}. Therefore $\D_{Alam}$ can be emebdded into $M_1(\Hv)=\Hv$ and $\D_{ort}$ into $M_2(\Hv)$.
\end{exam}

%******************************************************************************%
\subsection{Embedding space-time lattice codes into $M_k(\Hv)$}

We have given in Corollary \ref{cor:cond} the conditions for a division algebra $\Ac$ of index
$n$ to be embedded into $M_{n/2}(\Hv)$. To obtain a space-time lattice code, we need to select a discrete subset of
$\Ac$, namely one of its orders. We denote by $\Oc_K$ the ring of integers of $K$, and similarly by
$\Oc_E$ the ring of integers of $E$.
\begin{definition}\label{centerorder}
An $\OO_K$-order $\Lambda$ in $\A$ is a subring of $\mathcal{A}$, having the same identity
element as $\mathcal{A}$, and such that $\Lambda$ is a finitely generated module over $\OO_K$
and generates $\mathcal{A}$ as a linear space over $K$.
\end{definition}

This choice is motivated by the following example:

\begin{exam} \label{naturalorder}
Let $E/K$ be a cyclic extension of algebraic number fields and
$(E/K,\sigma,\gamma)$ be a cyclic division algebra,
with $\gamma\in K^*$  an algebraic integer. The ${\cal O}_K$-module
$$
\Lambda={\cal O}_E\oplus u {\cal O}_E\oplus\cdots\oplus u^{n-1}{\cal O}_E
$$
is a subring of the cyclic algebra $(E/K,\sigma,\gamma)$.
We refer to this ring as the {\it natural order} \cite{HLL}. Most space-time lattice codes built from
division algebras \cite{SRS,BORV} have been further restricted to this natural order.
\end{exam}

In theoretical considerations we will later mostly consider $\OO_K$-orders (where $K$ is the center) but the connection to coding theory is more visible if we consider $\OO_K$-orders as $\Z$-modules.
\begin{definition}\label{zorder}
A $\Z$-order $\Lambda$ in $\A$ is a subring of $\mathcal{A}$, having the same identity
element as $\mathcal{A}$, and such that $\Lambda$ is a finitely generated module over $\Z$
and generates $\mathcal{A}$ as a linear space over $\Q$.
\end{definition}

The ring $\Z$ is a principal ideal domain and therefore a $\Z$-order is not only finitely generated as a $\Z$-module, but it also has a $\Z$-basis. This basis is also a $\Q$-basis for the algebra $\mathcal{A}$. In particular a $\Z$-basis of an order in $\mathcal{A}$ has $ \mathrm{dim}_{\Q}(\mathcal{A})$ elements.

\begin{remark}
The ring $\OO_K$ is a finitely generated $\Z$-module. It is also known that $K$ is generated as a linear space over $\Q$. These results reveal that
any $\OO_K$-order is also a $\Z$-order.
\end{remark}

Let us again consider a general division algebra $\Ac$ having a center $K$, where
$[K:\Q]=m$, and let $\psi$ be the embedding of $\Ac$ defined in (\ref{psikuvaus}).

\begin{proposition}\label{alamring}
Let $\Lambda$ be a $\Z$-order of $\Ac$.
Then $\psi(\Lambda)$ is a $mn^2$ dimensional lattice in $M_{mn}(\C)$.  If
$$
\{a_1,\dots, a_{mn^2}\}
$$
is a $\Z$-basis of the order $\Lambda$, then
$$
\{\psi(a_1), \dots, \psi(a_{mn^2})\}
$$
is a $\Z$-basis of the lattice $\psi(\Lambda)$.

For any non-zero element of the order $\Lambda$, we have
$$
\mindet{\psi(\Lambda)}\geq 1.
$$
In particular $\psi(\Lambda)$ is a space-time lattice code that has the NVD
property (see Definition \ref{def:NVD}) and dimension rate $mn^2/mn= n$.
\end{proposition}
\begin{proof} The $\Z$-basis of $\Lambda$ has $\mathrm{dim}_{\Q}(\Ac)$ elements.
We have that $\Ac$ is of index $n$ and thus degree $n$, so it is of dimension
$n^2$ over  the center $K$. The center $K$ on the other hand is an $m$-dimensional $\Q$-vector space.
Overall we get that $\mathrm{dim}_{\Q}(\Ac)=mn^2$. Let us now consider a $\Z$-basis  $\{a_1,\dots, a_{mn^2}\}$
of $\Lambda$. While it is clear that the set $\{\psi(a_1), \dots, \psi(a_{mn^2})\}$ does generate $\psi(\Lambda)$,
it is not directly obvious that $\psi(a_1), \dots, \psi(a_{mn^2})$ are linearly independent over $\R$.
For this result and for the claim on $\mindet{\psi(\Lambda)}$, we refer the reader to \cite{Bayer}.

According to Definition \ref{def:R1}, the dimension rate $R_1$ for the code $\psi(\Lambda)$ is given by
$$
R_1=\frac{\dim_\RR(\psi(\Lambda))}{nm}=\frac{mn^2}{nm}=n
$$
dimensions per channel use.
\end{proof}

\begin{remark}
Due to the above connection between an order and a lattice, we may equally call a lattice code an \emph{order code}.
\end{remark}

If we now concentrate on codes that are embeddable into $M_k(\HH)$, we need to restrict to a $K$-central division algebra $\Ac$ of index $n$, where $K$ is totally real and all the infinite places are ramified.
We then get from (\ref{eq:embed}) an embedding
\[
\psi:\Ac \hookrightarrow \diag(M_{n/2}(\Hv)^m)\subset \diag(M_n(\CC)^m).
\]
By taking an order $\Lambda \subset \Ac$, we get a lattice code
$$
\psi(\Lambda)=\Z A_1\oplus\cdots \oplus \Z A_{mn^2} \subset M_{nm}(\C),
$$
where $A_i \in M_{nm/2}(\Hv)$, $i=1,\ldots mn^2$, forms a $\ZZ$-basis of the lattice.
Its dimension rate is similarly $n$.
It is clear that forcing a space-time code to be embedded in $M_{n/2}(\Hv)$
imposes an extra constraint. The next result characterizes this constraint in terms of the dimension rate.

\begin{proposition}\label{hamilrate}
Let us suppose that we have a lattice space-time code
$\Cc\subset M_{k}(\C)\cap M_{k/2}(\HH)$, where $k$ is even. We then have that
$$
\dim_{\R}(\Cc)\leq k^2.
$$
Consequently, the dimension rate $R_1$ of $\Cc$ as given in Definition \ref{def:R1} is at most $k$.
\end{proposition}

\begin{proof}
We can see that, as a subspace in $M_2(\C)$, the ring of Hamiltonian quaternions has degree $4$.
Each matrix in $M_{k/2}(\Hv)$ consist of
$(k/2)^2$ freely chosen $(2\times 2)$ blocks that have the inner structure of Hamiltonian quaternions. Therefore we have
$$
\mathrm{dim}_{\R}(M_{k/2}(\Hv))=4\left(\frac{k}{2}\right)^2=k^2.
$$
\end{proof}

If we compare the rate  $n$ of $\psi(\Lambda)$ with this result, we get $n$ versus $nm$, where $m=[K:\QQ]$.
There is thus a trade-off between fast decodability and rate. However, by choosing the center of the algebra $\Ac$ to be $\Q$,
we can  meet the optimal dimension rate of Proposition  \ref{hamilrate}.

\begin{remark}
We warn the reader here. The theory developed so far is not explicit in a sense that while it does give us a good description
of how to construct the needed division algebras (see Proposition \ref{CM}), we have not given  an explicit method to produce the embedding (\ref{psikuvaus}). In particular, we have no guarantee that the left regular representation would have anything to do with the embedding (\ref{psikuvaus}).  In Section \ref{explicit} and the following parts of the paper,   we will show that there are methods to  overcome this problem and that the left regular representation can work  as a good  starting point.

\end{remark}
%******************************************************************************%
%
% BOUNDS
%
%******************************************************************************%

\section{Bounds and existence results for matrix lattices in $M_k(\HH)$}
\label{bounds}

So far, we have given conditions for a division central $K$-algebra $\Ac$ to be embedded into $M_k(\HH)$ 
and shown how to obtain fast-decodable space-time lattice codes from orders of $\Ac$. 
In this section we are going to give bounds and existence results for such codes, taking into
account an extra code design criterion, namely the normalized minimum determinant of a lattice
code.

%**************************************************************************%
\subsection{Normalized minimum determinant of an order code}
\label{mindetorder}
The minimum determinant $\mindet{\Cc}$ is a widely used concept to predict the performance of a finite space-time code $\Cc$, since it determines its coding gain.
In order to compare two finite space-time codes $\Cc_1,\Cc_2 \,\in M_n(\C)$,
one must first  check that
\begin{itemize}
\item
both codebooks have equal number of elements: $|\Cc_1|=|\Cc_2|$ and
\item
both codes are scaled so that the maximum power used is equal:
$\max\{||A||_F^2 \,|\, A\in C_1\}= \max\{||B||_F^2 \,|\, B\in C_2\}$.
\end{itemize}

In the case of infinite lattice codes, due to the discreteness of the set, a non-zero minimum determinant automatically yields the NVD property. Among two NVD codes using the same maximum power, the one with  higher minimum
determinant  will have better coding gain for the infinite lattice, and will thus provide us with a
bound on the coding gain of any finite constellation carved from it. Now given an infinite  space-time
lattice code $\Cc$, a number $R$ of codewords, and a fixed power constraint,
there are different ways to pick a finite constellation that may lead to
different coding gains. 

The two most typical encoding methods are linear dispersion encoding (cf. the discussion underneath Equation \eqref{eq:CG}) and spherical encoding. These encoding methods usually result in different  constellation shaping, that can be either  
cubic (more generally orthogonal) shaping, provided the lattice is orthogonal to start with, or spherical shaping. The two possible shapes  are described below in more detail.

{\bf Spherical shaping.}
Just as for Gaussian channels, the most energy efficient way to choose codewords from a given lattice is to use spherical shaping. This means
that we choose the needed number of lowest energy codewords from the space-time lattice code $\Cc$ and then scale the finite code $\Cc(r)$ given by
\begin{equation}\label{spherical}
\Cc(r)=\{\,A \,|\, A\in \Cc, ||A||_F\leq r\}\subset \Cc
\end{equation}
to meet the power constraint, where $r$ depends on the number $R$ of wanted
codewords. For large code sizes, this approach will roughly give lattice points inside a $K$-sphere, where $K$ is the rank of the code lattice (=number of dispersion matrices).

To fairly compare  two finite codes $\Cc_1(r)$  and $\Cc_2(r)$, one should first
scale them so that both the lattices have a fundamental parallelotope of volume 1. 	
Since we consider a space-time lattice code $\Cc\in M_n(\CC)$, to define its volume we first map it to $\RR^{2n^2}$ via $\alpha$, yielding the lattice
$\alpha(\Cc)$ whose basis is $\{\alpha(B_1),\ldots,\alpha(B_K)\}$, obtained from the basis $\{B_1,\ldots,B_K\}$ of $\Cc$. The generator matrix $M$ of
$\alpha(\Cc)$ is $M=(\alpha(B_1),\ldots,\alpha(B_K))$, where $\alpha(B_i)$ are column vectors, and we define the measure (or volume) $m(\Cc)$ of the fundamental parallelotope of the space-time
lattice $\Cc$ by
\[
m(\Cc)^2=\det(MM^T)=\det(\left(\Re \tr(B_iB_j^\dagger)\right)_{1\le i,j\le K}.
\]
To combine the notion of minimum determinant with that of scaling the volume
of the lattice to evaluate the performance of finite constellations, we use
the notion of {\em normalized minimum determinant} $\delta(\Cc)$, obtained  by first scaling the lattice $\Cc$ to have a unit size fundamental parallelotope and then taking the minimum determinant of the resulting scaled lattice. A simple computation proves the following.
\begin{lemma}\label{scale}
Let $\Cc$ be a $K$-dimensional space-time lattice
in $M_n(\C)$. We then have that
$$
\delta(\Cc) =\mindet{\Cc}/(m(\Cc))^{n/K}.
$$
\end{lemma}

The normalized minimum determinant predicts which lattice is likely to
produce the finite codes with the biggest minimum determinants, while using spherical shaping.

{\bf Cubic shaping.}
 We also consider another kind of shaping, called cubic or orthogonal shaping.
\begin{definition}
We say that a  space-time lattice $\Cc$ in $M_n(\C)$ is orthogonal or rectangular
if the corresponding real lattice $\alpha(\Cc)$ has a basis that is orthogonal according to the normal inner product of the space
$\R^{2n^2}$. If each of of the basis vectors are of  equal length, we say that $\Cc$ is orthonormal.
\end{definition}

When the lattice is orthogonal, there is no point of employing spherical shaping (\ref{spherical}), for we get the same result by using simple linear dispersion encoding (see the remark in the end of this section) as described after Equation  \eqref{eq:CG}.
 
One can get bounds for the normalized minimum determinant also  in the case of cubic shaping,
as for example:
\begin{proposition}\cite{jyrkiroopeAAECC}\label{optimality}
Let us suppose that $\Cc$  is an orthogonally shaped 16-dimensional space-time lattice code in $M_4(\CC)$.
We then have that
$$
\delta(\Cc) \leq \frac{1}{16}=0.0625.
$$

\end{proposition}

In the particular case where $\Cc$ is an order code, that is
$\Cc=\psi(\Lambda)$, with $\Lambda$ an order of an index $n$ division algebra
$\A=(E/K, \sigma, \gamma)$ and $[K:\Q]=m$, we know from Proposition
\ref{alamring} that $\psi(\Lambda)$ is an $mn^2$-dimensional lattice in
$M_{mn}(\C)$ with $\mindet{\psi(\Lambda)}= 1$, so that
\[
\delta(\psi(\Lambda)) = 1 /(m(\Cc))^{1/n}
\]
and the normalized minimum determinant only depends on the volume of the fundamental parallelotope of the order code.

\begin{remark} Note that the fact whether one uses linear dispersion encoding (i.e., a symmetric coefficient set) or spherical shaping (i.e., an optimized coefficient set) has nothing to do with the shape of the original lattice. Even though the lattice is not orthogonal, we can employ both encoding methods. If the lattice is not badly skewed, then the difference between the two methods is usually not very big, whereas for highly skewed lattices one may see a gap of several dBs.

For orthogonal lattices, both methods will give the same result, provided that the target constellation size is suitable for a symmetric coefficient set to start with. 
\end{remark}

%**************************************************************************%
\subsection{Bounds and existence results}

Since the normalized minimum determinant of an order code only depends on
the volume of its fundamental parallelotope, one may wonder whether,
given a center $K$, it is possible to find the smallest volume an order inside any division algebra of a given index $n$ can have.

To answer this question, we first further characterize the volume of the order
by connecting it to an invariant of the order.

\begin{proposition}\cite{Bayer}\label{volume}
Let $\Lambda$ be a $\ZZ$-order in $\A$ and let $\psi$ be the embedding (\ref{psikuvaus}). We then have that 
\begin{align*}
m(\psi(\Lambda))= \sqrt{|d(\Lambda/\Z)|},
\end{align*}
where $d(\Lambda/\Z)$ is the $\Z$-discriminant of the order $\Lambda$ (see \cite{R,HLRV2} for an exact definition),
and further that
$$
\delta({\psi(\Lambda)})=\left(\frac{1}{|d(\Lambda/\Z)|}\right)^{1/2n}.
$$
\end{proposition}

Clearly the smaller the absolute value of the $\Z$-discriminant of an order is, the greater the normalized minimum determinant will be.

Inside a given algebra the $\Z$-orders having the smallest possible discriminant are called  \emph{maximal orders}. All the maximal orders of a given division algebra share the same discriminant.

While each $\OO_K$-order is also $\Z$-order, the opposite does not have to be true. However if a $\Z$-order $\Lambda$ also is an $\OO_K$-module, it is an
$\OO_K$-order and its $\OO_K$-discriminant $d(\Lambda/\OO_K)$ is related to
its $\Z$-discriminant by the following transitivity formula:
\begin{lemma}\label{disktorni}
Let $\A$ be a $K$-central division algebra of index $n$ and let $\Lambda$ be an $\OO_K$-order.
If $\Lambda$ is a $\Z$-order in $\A$, then
$$
d(\Lambda/\Z)=\Nc_{K/\Q}(d(\Lambda/\OO_K)) d(\OO_K/\Z)^{n^2},
$$
where $d(\OO_K/\Z)$ is just the usual number field discriminant of the extension $K/\Q$.
\end{lemma}

To summarize, we have just shown that the normalized determinant
\[
\delta(\psi(\Lambda)) = 1 /(m(\Cc))^{1/n}
\]
is given by
\[
\delta(\psi(\Lambda)) =\left(\frac{1}{|\Nc_{K/\Q}(d(\Lambda/\OO_K)) d(\OO_K/\Z)^{n^2} |}\right)^{1/2n}.
\]
This reveals that we only have to consider the term
$$
\Nc_{K/\Q}(d(\Lambda/\OO_K))
$$
as $d(\OO_K/\Z)^{n^2}$ is fixed (when $K$ is fixed). The $\OO_K$-discriminant $d(\Lambda/\OO_K)$ is an ideal in $\OO_K$, but
$\Nc_{K/\Q}(d(\Lambda/\OO_K))$ can be seen as an element in $\Z$. Therefore we can discuss the size
of ideals of $\mathcal{O}_K$. By this, we mean that ideals are ordered by the
absolute values of their norms to $ \Q$. For example, if
$\mathcal{O}_K=\Z[i]$, we say that the prime ideal generated by $2+i$ is
smaller than the prime ideal generated by $3$, because they have norms 5 and 9,
respectively.

We are now ready to state the bounds that characterize the best order codes
in terms of normalized minimum determinant. The hypotheses take into account
that the order code can be embedded into $M_k(\HH)$, for some $k$.

In the following, we use the notation $2\mid\mid n$ which means that $2$ divides $n$, but $4$ does not.

\begin{proposition}\label{realinfinitebound}
Let $\mathcal{A}$ be a  $K$-central division algebra of index
$n$, $2 \mid n$, where $K$ is a totally real number field, and let $P_1\leq P_2$ be a pair of smallest primes in $K$. Let us suppose that all the infinite primes are ramified in $\mathcal{A}$.

If  $2\mid \mid n$ and  $ 2\mid [K:\Q]$, then the minimum discriminant of
$\mathcal{A}$ is
$$
(P_1P_2)^{k(k-1)}.
$$

If $4\mid n$ then the minimum discriminant of
$\mathcal{A}$ is
$$
(P_1P_2)^{n(n-1)}.
$$
If $2\mid\mid n$ and  $2\nmid [K:\Q]$,  then the minimal discriminant of
$\mathcal{A}$ is
$$
P_1^{n(n-1)}P_2^{k(k-1)}.
$$
\end{proposition}
\begin{proof} The proof with related background as well as more
general bounds can be found in Appendix.
\end{proof}

\begin{exam}\label{boundexample}
Consider the question of building a $16$-dimensional lattice code in $M_4(\C)$ with the best achievable normalized minimum determinant. The order code
$\psi(\Lambda)$ gives an $mn^2$-dimensional lattice code in $M_{nm}(\C)$ for any order $\Lambda$. To have $nm=4$ and $mn^2=16$, the only option is to choose
$m=1$ and $n=4$. According to Proposition \ref{infinitebound}, we have that the smallest possible discriminant for a $\Q$-central division algebra of index $4$ is $2^{12} \cdot 3^{12}$. Let us now suppose that
$$
\A=(E/\Q,\sigma, \gamma)
$$
is the algebra having a maximal order $\Lambda$ with the promised discriminant. According to Proposition \ref{volume} we have that
$$
m(\psi(\Lambda))=6^6 \ \mathrm{ and }\  \delta (\psi(\Lambda))=\left(\frac{1}{6^{12}}\right)^{\frac{1}{8}}=0.068...
$$
Proposition \ref{realinfinitebound} tells us that
we can achieve this bound even with a $16$-dimensional lattice in
$M_4(\C)\cap M_2(\HH)$.
\end{exam}

In \cite{asykonstru}, the authors managed to build a $16$-dimensional lattice code IA-MAX in $M_4(\C)$ having a normalized minimum determinant equal to $0.1361...$. We however conjecture that $0.068....$ is the best possible minimum determinant
for a lattice in $M_4(\C)\cap M_2(\HH)$.

%******************************************************************************%
%
% CONSTRUCTIONS
%
%*****************************************************************************%
\section{Explicit construction methods}\label{explicit}

So far our study has been mostly theoretical. No explicit constructions resulting from the mapping $\psi$ \eqref{psikuvaus} have yet been given. We have only proved that the afore described matrix lattices with NVD exist. 
Let us now suppose that we have a $K$-central division algebra $\D=(E/K, \sigma, \gamma)$, where $[K:\Q]=m$ and $[E:K]=n$. There exist $m$ $\Q$-embeddings $\beta_i$  from $K$ to $\C$. For each $\beta_i$ we can find such an embedding $\sigma_i:$  $E\hookrightarrow \C$ that $\sigma_i|_{K}=\beta_i$. Let us now suppose that $\{\sigma_1,\dots, \sigma_m\}$  is a set of representatives of embeddings $\beta_i$.

By using the left maximal representation we get an embedding
$\phi:\D\hookrightarrow M_n(E) \subseteq M_n(\C)$. Let us suppose that $a$ is an element of $\D$ and $A$ is the corresponding matrix $\phi(a)$. We then get a mapping 
\begin{equation}\label{psistarkuvaus}
\psi^*: \D\rightarrow M_{n\times nm}(\C) 
\end{equation}
which is defined by
$$
a\mapsto \mathrm{diag}(\sigma_1(A),\dots, \sigma_m(A)).
$$

We now have the following explicit version of the previously defined embedding (\ref{psikuvaus}).
\begin{proposition}\label{psistar}
 Let us suppose that $\Lambda$ is a $\Z$-order in $\D$ and  that $\psi^*$ is the  embedding \eqref{psistarkuvaus} defined above.
Then $\psi^*(\Lambda)$ is a $mn^2$ dimensional lattice in $M_{mn\times nm}(\C)$. For any non-zero element of the order $\Lambda$ we have
$$
det_{m}(\psi^*(a))\geq 1.
$$
\end{proposition}

However, in general we might loose the connection between the volume of the fundamental parallelotope of the order code $\psi^*(\Lambda)$ and 
the $\Z$-discriminant of $\Lambda$. However if we can choose the left regular representation and the embeddings $\sigma,\dots,\sigma_m$
correctly we have the following. Let us suppose that we have such a center $K$ and an index $n$ division algebra $\A$ that
$$
\Ac\otimes_{\Q}\RR \cong M_{n/2}(\Hv)^{\omega} \times M_{n}(\R)^{r_1 -\omega} \times G(\C)^{r_2}.
$$

\begin{proposition}\label{psistarvol}
Let us suppose that $\Lambda$ is a $\ZZ$-order in $\A$ and  that $\psi^*$ is the previously defined embedding. If we can choose
$\sigma_1,\dots, \sigma_m$ and a left maximal representation $\phi$ so that 
$$
\psi^*(\Lambda)\subset\diag(M_{n/2}(\Hv)^{\omega} \times M_{n}(\R)^{r_1 -\omega} \times G(\C)^{r_2}),
$$
we  get
\begin{align*}
m(\psi^*(\Lambda))= \sqrt{|d(\Lambda/\Z)|}
\end{align*}
and
$$
\delta({\psi^*(\Lambda)})=\left(\frac{1}{|d(\Lambda/\Z)|}\right)^{1/2n}.
$$
\end{proposition}
\begin{proof}
Under the assumption that the embeddings and the maximal representation are chosen as presented the proof of these claims is verbatim the same as for 
 Proposition \ref{volume} and can therefore found from \cite{Bayer}.
\end{proof}

Unfortunately in the proof of the following proposition we have to use some notions not defined in this paper.
\begin{proposition}\label{thm:conjugation}
Let us suppose we have an index $n$ $\Q$-central division algebra and let $\phi$ denote the left regular representation. If we have such a real matrix $M$ that
$$
M\phi(\D)M^{-1} \subseteq M_{n/2}(\HH),
$$
then
$$
\delta({M\phi(\Lambda)M^{-1}})=\left(\frac{1}{|d(\Lambda/\Z)|}\right)^{1/2n}.
$$
\end{proposition}
\begin{proof}
We will give the proof in the case where the index is $2$. The generalization is obvious and we will meet  all the needed ideas already in this simplest case.

Let us suppose that  $\phi(\Lambda)$  has a  $\Z$-basis $\{A_1, A_2,A_3, A_4\}$. We denote $B_i=MA_iM^{-1}$   and set $\mathcal{B}=\{(B_1, \dots, B_4\}$.  We can  flatten the matrix $B_i$ into a $4$-tuple $L(B_i)$  by first forming a vector of length $4$ out of the entries of $A_i$ (e.g. row by row).
The following identities are now easily seen
\begin{equation}\label{eq:L(A)H}
L(B_i)L(B_j)^T= \mathrm{Tr}(B_i B_j^T)
\end{equation}
and
\begin{equation}\label{eq:L(A)}
L(B_i)L(B_j^T)^T= \mathrm{Tr}(B_i B_j).
\end{equation}

The Gram matrix of the lattice $M\phi(\Lambda)M^{-1}$ is
$$
G=(\Re (\mathrm{Tr}(B_i B_j^{\dagger})))_{i,j=1}^4.
$$
Both $B_i$ and $B_j^{\dagger}$ do have Alamouti structure and therefore so does also    $B_i B_j^{\dagger}$.
 This reveals that $\mathrm{Tr}(B_i B_j^{\dagger}) \in \R$ and we can omit taking the real part from the Gram matrix.

According to Equation \eqref{eq:L(A)H} we can  now write
$$
G=(L(B_i) L(B_j^*)^T)_{i,j=1}^4=L(\mathcal{B})L(\mathcal{B})^{\dagger},
$$
 where the rows of the     $4 \times 4$ matrix $L(\mathcal{B})$ consist of the vectors $L(B_i)$. A simple permutation of the columns and elementary properties of determinants give us that
$$
|\mathrm{det}(L(\mathcal{B}))\mathrm{det}(L(\mathcal{B})^{\dagger})|=
$$
$$
|\mathrm{det}(L(\mathcal{B}))\mathrm{det}(L(\mathcal{B})^T)|=|\mathrm{det}(L(\mathcal{B}))\mathrm{det}(L(\mathcal{B'})^T)|,
$$
where $L(\mathcal{B'})$ is a matrix with the rows       $L((B_i)^T)$.
According to Equation \eqref{eq:L(A)} we now have
%\medskip\medskip
$$
 L(\mathcal{B})L(\mathcal{B'})^T =(\mathrm{Tr}(MA_i A_j M^{-1}))_{i,j=1}^4.
$$
A general  result on matrix  traces   tells us that $\mathrm{Tr}(XCX^{-1})=\mathrm{Tr}(C)$ for any matrices $C$ and $X$. This result combined with the
definition of the  discriminant now gives us that
$$
L(\mathcal{B})L(\mathcal{B'})^T =(\mathrm{Tr}(MA_i A_j M^{-1}))_{i,j=1}^4=
$$
$$
 (\mathrm{Tr}(A_i A_j))_{i,j=1}^4 =\sqrt{d(\Lambda/\Z)}.
$$

\end{proof}

\begin{exam}
Consider from (\ref{eq:Dalam}) the division algebra
\[
\D_{Alam}=(\Q(i)/\Q,\sigma, -1),
\]
which has index 2 and center $\Q$. The field $\QQ$ has only one infinite
place $\infty$ and according to Proposition \ref{CM} it is ramified in the algebra $\D_{Alam}$.
We thus have an embedding $\D_{Alam}\hookrightarrow \Hv$ given by
\eqref{eq:embed}. If we choose a $\Z$-order $\Lambda$ in $\D_{Alam}$,
$\psi(\Lambda)\subset \HH\subset M_2(\C)$ is a $4$-dimensional lattice code.
 
Here the left regular representation directly gives us an explicit version (see (\ref{psistarkuvaus}) and Proposition \ref{psistar}) of this mapping. As demonstrated in the beginning of the paper, it also gives us a fast-decodable code.
\end{exam}

\begin{exam}
Let us consider the example we gave in the very beginning of the paper.
The cyclic algebra
$$
\D_{ort}=(\Q(i, \sqrt{2})/\Q(\sqrt{2}), \sigma,-1),
$$
 is an index $2$ division algebra with center $\Q(\sqrt{2})$.  Here $\sigma$ is simply the complex conjugation. The general theory tells us that 
 $\D_{ort}$ can be embedded  into $M_2(\HH)$. 
 
Again the mapping from Proposition \ref{psistarkuvaus} will directly give us an explicit version of the  embedding in \eqref{eq:embed}. 
The field $\Q(\sqrt{2})$  has two $\Q$-embeddings $\beta_1$, $\beta_2,$ where $\beta_1(\sqrt{2})=\sqrt{2}$ and 
$\beta_2(\sqrt{2})=-\sqrt{2}$. The corresponding $\Q$-embeddings $\sigma_1$ and $\sigma_2$ are defined by the equations
$\sigma_1=id$, $\sigma_2(i)=i$  and   $\sigma_2(\sqrt{2})=-\sqrt{2}$ (or equivalently $\sigma_2(\zeta_8)=-\zeta_8$). The natural order $\Lambda$ consists of elements
$x=a_1+a_2\zeta_8 +u a_3 + u\zeta_8 a_4$, where $a_i\in \Z[i]$. The left regular representation now gives us  
$$
\alpha(x)=
\begin{pmatrix}
a_1+a_2 \zeta_8& -a_3^*-a_4^*\zeta_8^*\\
a_3+a_4\zeta_8 & a_1^*+a_2^*\zeta_8^*
\end{pmatrix}.
$$
It is then an easy task to see that 
$$
\sigma_2(\alpha(x))=
\begin{pmatrix}
a_1-a_2\zeta_8 &-a_3^*+a_4^*\zeta_8^*\\
 a_3-a_4\zeta_8 &a_1^*-a_2^*\zeta_8^*
\end{pmatrix}.
$$
In particular both $\alpha(x)$ and $\sigma_2(\alpha(x))$ are elements in $\HH$ and Proposition \ref{psistarvol} can be applied.
These results reveal that the example code we gave in the beginning of the paper was just an instance of the general theory developed above.
\end{exam}

\begin{remark}
These two examples may give us a little too rosy picture of the power of our theory. In both cases, the embedding in Proposition \ref{psistar} exactly imitated the embedding  \eqref{eq:embed}. On top of that this representation also led to codes with reduced decoding complexity. However, we do not have any guarantee that either of these things will happen more generally. It heavily depends on the chosen maximal subfield, non-norm element and even on the chosen generator of the Galois group. In  Sections \ref{sec:codes} and \ref{sec:codes6} we will  meet
situations where the left regular representation does not directly give us the required embedding even when the division algebra has the  correct algebraic structure. Yet, in all these cases a simple manipulation applied after the left regular representation will give us an embedding to the matrix ring of quaternions and codes that have reduced decoding complexity.
While this may seem to be  accidental, there are some underlying  algebraic principles that explain the sudden ``luck'' we encounter, see Section \ref{sec:furthergen}.
\end{remark}

\section{Fast-decodable $4\times 2$ MIDO codes}
\label{sec:codes}

So far, we have developed an algebraic theory of fast-decodable codes through different 
characterizations and bounds. We are now finally putting our theory into use to 
give a few different coding strategies that lead to fast-decodable codes. 
We start with MIDO codes for 4 Tx antennas, with the following properties:
\begin{itemize}
\item 
They are 16-dimensional lattices in $M_4(\C)$.
\item 
They satisfy the NVD property.
\item 
Their decoding complexity ranges from $|S|^{10}$ to $|S|^{16}$ 
when a real alphabet of size $|S|$ is used.
\end{itemize}

%******************************************************************************%
\subsection{A family of fast-decodable MIDO codes with $\Q$ as a center}
\label{family}

We give here an example of a MIDO code built following step by step the theory 
developed so far. The starting point is  to consider a division algebra that
can be embedded into $M_2(\HH)$ via the embedding $\psi$ (\ref{psikuvaus}). 
According to Section \ref{embed} and Proposition \ref{CM}, we consider a 
$\Q$-central division algebra $\mathcal{A}=(E/\mathbb{Q},\sigma,\gamma)$ of 
index $4$, where $E$ is a CM field and $\gamma$ a negative non-norm element, 
namely
\begin{enumerate}
\item[c1)] $[E:\mathbb{Q}]=4$,
\item[c2)] $\gamma,\gamma^2\notin \Nc_{E/\QQ}(E^*)$,
\item[c3)] $\textrm{Gal}(E/Q)=\langle\sigma\rangle$ with $\sigma^2(f)=f^*$, 
where $f^*$ stands for the complex conjugate of $f$, and
\item[c4)] $\gamma<0$.
\end{enumerate}

One instance of such an algebra is
\[
\D_{mido}=(\Q(\zeta_5)/\QQ ,\sigma, -8/9),
\]
where $\sigma$ is given by $\sigma(\zeta_5)=\zeta_5^3$. 
The prime $2$ is totally inert in the extension $\Q(\zeta_5)/\QQ$ and therefore 
\cite[Lemma 11.1]{HLRV2} $\D_{mido}$ is a division algebra. 

Let $\OO_E=\Z w_1 \oplus \Z w_2 \oplus \Z w_3 \oplus \Z w_4$ be the ring of 
algebraic integers of $E$.
The left representation $\phi^*$ of $\D_{mido}$ now yields
\begin{equation}\label{eq:matrixwconj}
\left(
\begin{array}{rrrr}
y_1 & \gamma\sigma(y_4)& \gamma y_3^* & \gamma\sigma(y_2)^*\\
y_2 & \sigma(y_1) & \gamma y_4^* & \gamma\sigma(y_3)^*\\
y_3 & \sigma(y_2) &  y_1^*  & \gamma\sigma(y_4)^*\\
y_4 & \sigma(y_3) & y_2^*  & \sigma(y_1)^*\\
\end{array}
\right),
\end{equation}
where  $y_i=y_i(g_{4i-3},g_{4i-2},g_{4i-3},g_{4i}) = g_{4i-3}w_1+g_{4i-2}w_2 +g_{4i-3}w_3+ g_{4i}w_4$ and  $g_{4i-j}\in \Q$ for $i=1,2,3,4$, $j=0,1,2,3$. If we pick up an order $\Lambda$ 
from $\D_{mido}$, then $\psi^*(\Lambda)$ is a $16$-dimensional lattice code with the NVD property 
from Proposition \ref{psistarvol}. 

We can prove that the discriminant of this algebra   meets the bound of Proposition \ref{realinfinitebound}, but even if we choose a maximal order from this algebra there is no  guarantee (because we have not fulfilled the conditions of Proposition \ref{psistarvol} yet) that this small discriminant would result into good normalized minimum determinant. 

This is because we now face here,   for the first time, the problem that the embedding $\psi^*$ from Section \ref{explicit} does not directly give us an embedding into $M_2(\Hv)$, although Proposition \ref{psistar} promises that such an embedding exists.  Luckily, we can perform a series of simple manipulations starting from the left regular representation that will transform the code matrices into a correct form and at the same time will recover the connection between the discriminant of the algebra and the normalized minimum determinant of the lattice.

After swapping
\begin{enumerate}
\item $y_2$ and $y_3$,
\item the 2nd and the 3rd column, and
\item the 2nd and the 3rd row,
\end{enumerate}
we get the matrix
\begin{equation}\label{eq:matrixafterswap}
\left(
\begin{array}{rrrr}
y_1 &  \gamma y_2^* & \gamma\sigma(y_4)&\gamma\sigma(y_3)^*\\
y_2 & y_1^*   & \sigma(y_3) & \gamma\sigma(y_4)^*\\
y_3  & \gamma y_4^* & \sigma(y_1)& \gamma\sigma(y_2)^*\\
y_4 &  y_3^*  & \sigma(y_2) & \sigma(y_1)^*\\
\end{array}
\right).
\end{equation}

Next we perform the following energy 
balancing transformation by distributing the effect of $|\gamma|$ 
more evenly. By 
denoting $r=|\gamma|^{1/4}$, we finally get a code consisting of matrices of the desired type:
\begin{eqnarray}\label{eq:finalmatrix}
&& X_{FD}(y_1,y_2,y_3,y_4)\\ \nonumber
&=&\left(\begin{array}{rrrr}
y_1&-r^2 y_{1}^*&-r^3\sigma(y_4)&-r\sigma(y_3)^*\\
r^2y_2&y_1^*&r\sigma(y_3)&-r^3\sigma(y_4)^*\\
r y_3&-r^3y_{3}^*&\sigma(y_1)&-r^2\sigma(y_2)^*\\
r^3 y_{3}&r y_{2}^*&r^2\sigma(y_{1})&\sigma(y_1)^*\\
\end{array}\right).
\end{eqnarray}

The  minimum determinant of the code stays unchanged since the above 
transformation is actually just a conjugation by a real matrix $M$. Let us now suppose that we have 
a maximal order $\Lambda$ of the algebra $\D_{mido}$ (such an order can  be found by using the computer 
algebra system Magma \cite{Mag}). Now the new code obtained from this maximal order is $M\psi^*(\Lambda)M^{-1}$, and a direct calculation reveals that this code lattice meets the normalized minimum  determinant bound $\delta(\phi(\Lambda))=0.068...$ (cf. Propositions \ref{psistarvol}, \ref{thm:conjugation}, \ref{realinfinitebound}, and Example \ref{boundexample}).   

To make the code suitable for PAM modulation, we further describe a modified version 
of this code that will have an almost rectangular shaping. 
The ring of algebraic integers in $\Q(\zeta_5)$ also has a  $\Z$-basis 
$\{1-\zeta,\zeta-\zeta^2,\zeta^2-\zeta^3,\zeta^3-\zeta^4\}$, where we have abbreviated $\zeta_5=\zeta$. The elements in the code 
matrix (\ref{eq:finalmatrix}) now become, after further restricting the coefficients 
$g_i$ to $\Z$:
\begin{eqnarray*}
y_i'&=&y_i'(g_{4i-3},g_{4i-3},g_{4i-2},g_{4i}) \\
&=&g_{4i-3}(1-\zeta) +g_{4i-2}(\zeta-\zeta^2) +\\
&&+g_{4i-1}(\zeta^2-\zeta^3) + g_{4i}(\zeta^3-\zeta^4)
\end{eqnarray*}
and 
\begin{eqnarray*}
\sigma(y_i')&=&
g_{4i-3}(1-\zeta^3)+g_{4i-2}(\zeta^3-\zeta)\\
&&+g_{4i-1}(\zeta-\zeta^4)+g_{4i}(\zeta^4-\zeta^2).
\end{eqnarray*}
We get a set of matrices $X_{FD,A_4}(y_1',y_2',y_3',y_4')$ forming a 16-dimensional 
lattice code in $M_2(\HH)$. We note that the choice of $\gamma=-8/9$ prevents this order code from being a 
natural order. However, after multiplication by $9^4$, the resulting lattice code will be included 
in a natural order, thus inheriting the NVD property. The geometric structure of the code is relatively 
close to a Cartesian product of four  $A_4$-lattices (see \cite{CS}), therefore we call it the $A_4$ code. This code was also proposed for the DVB Consortium's \emph{Call for Technologies for DVB-NGH} \cite{DVB-CFT}.

The variables $g_{4i-j}$ in each of the $y_i'$ range over a certain PAM set, so that the code encodes 
overall $16$ independent PAM symbols. In other words, a PAM vector $(g_1,\dots,g_{16})$ is 
mapped into a $(4\times4)$ matrix
$$
\sum_{i=1}^{16} g_i B_i,
$$
where the basis matrices  $B_i$ of the code are
\begin{eqnarray*}
B_1&=&X_{FD,A_4}(y_1'(1,0,0,0),0,0,0),\\
B_2&=&X_{FD,A_4}(y_1'(0,1,0,0),0,0,0),\\
&\vdots&\\
B_{16}&=&X_{FD,A_4}(0,0,0,y_4'(0,0,0,1)).\\
\end{eqnarray*}
A direct calculation shows that
\[
\Re(\mathrm{Tr}(HB_i (HB_j)^{\dagger})=0
\]
for $1\leq i \leq 4$ and $5\leq j \leq 8$, 
where $H$ is a $(2\times 4)$ channel matrix. This is exactly the design criterion
of Subsection \ref{fastdec} described by the steps 1-2, yielding a complexity of $|S|^{12}$ for the code $A_4$.

We can perform yet another change of basis that will enable us to take advantage of the steps 3-4 described in Subsection \ref{fastdec}. The new basis
$$
\left\{1,\frac{\zeta+\zeta^{-1}}2,\frac{\zeta-\zeta^{-1}}2,\frac{\zeta^2-\zeta^{-2}}4\right\}
$$
 will result in a complexity $|S|^{10}$,  reduced by as much as 37.5\% from the full complexity  $|S|^{16}$ of a general MIDO code. However, it is not an integral basis, hence the minimum determinant is small  though still non-vanishing.

The resulting lattice has almost cubic shaping, but, due to the coding gain loss, the performance is approximately 1 dB worse than that of the $A_4$ version. The promised complexity reduction is due to the fact that the first two basis elements are real, while the last two are purely imaginary. Hence the relations given by the steps 1-4 in \ref{fastdec} are all satisfied.

\begin{remark} To the best of our knowledge, there is no guarantee that an integral basis consisting of $n/2$ real and  $n/2$ purely imaginary elements even exists.
\end{remark}

\begin{remark}
The matrix manipulations given in this section may also seem to have a somewhat \emph{ad hoc} feeling. Yet we will see in Sections \ref{sec:codes6} and \ref{sec:furthergen} that this strategy can be used  far more generally to give us embeddings to $M_k(\Hv)$.
\end{remark}

\begin{remark}
We also simulated the maximal order code from this algebra achieving the discriminant bound and the $A_4$ code under spherical shaping. Both codes had equally good performance, gaining almost 1 dB compared to the linearly dispersed $A_4$.  It seems that the $A_4$ code  did inherit the good performance of the optimal maximal order code.

\end{remark}

%***********************************************************************%
\subsection{MIDO codes from a bigger center through puncturing}

We now adopt a slightly different approach to the design problem of MIDO codes 
via puncturing of MIMO codes.
We start from the matrix (\ref{eq:matrix}) 
\[
\left(
\begin{array}{rrrr}
x_0 & i\sigma(x_3)& i\sigma^2(x_2) & i\sigma^3(x_1)\\
x_1 & \sigma(x_0) & i\sigma^2(x_3) & i\sigma^3(x_2)\\
x_2 & \sigma(x_1) & \sigma^2(x_0)  & i\sigma^3(x_3)\\
x_3 & \sigma(x_2) & \sigma^2(x_1)  & \sigma^3(x_0)\\
\end{array}
\right)
\]
and puncture it in two different ways.

Let us first repeat a remark made above, namely that
$\QQ(\zeta_5+\zeta_5^{-1})=\QQ(\sqrt{5})$ is a subfield of $\QQ(\zeta_5)$.
As a first puncturing, we restrict ourselves to elements in
$\QQ(\sqrt{5})$ instead of $\QQ(\zeta_5)$. Note that since
$\sigma^2(\zeta_5)=\zeta_5^4$, we further have
\[
\sigma^2(\zeta_5+\zeta_5^{-1})=\zeta_5^4+\zeta_5^{-4}=\zeta_5^{-1}+\zeta_5
\]
and thus $\QQ(\sqrt{5})$ is fixed by $\sigma^2$.
This yields a codebook $\Cc_1$ consisting of codewords of the form
\begin{equation}\label{eq:MIDO1}
X=
\frac1{\sqrt{5}}\left(
\begin{array}{rrrr}
x_0 & i\sigma(x_3)& ix_2 & i\sigma(x_1)\\
x_1 & \sigma(x_0) & ix_3 & i\sigma(x_2)\\
x_2 & \sigma(x_1) & x_0  & i\sigma(x_3)\\
x_3 & \sigma(x_2) & x_1  & \sigma(x_0)\\
\end{array}
\right).
\end{equation}
It is now enough to notice that we are working in the same field
extension as for the Golden code \cite{BRV}, meaning that we can use
the same shaping technique. Denote:
\begin{eqnarray*}
\theta&=& \frac{1+\sqrt{5}}{2},\\
\sigma(\theta) &=& \frac{1-\sqrt{5}}{2} = 1-\theta,\\
\alpha&=&1+i-i\theta,\\
\sigma(\alpha) &=& 1 + i - i \sigma(\theta).
\end{eqnarray*}
Every entry $x_j$ in the above matrix is now taking the form
\[
x_j=\alpha(a_j+b_j\theta),~j=0,1,2,3,
\]
where $a_j,b_j\in\Z[i]$ are chosen to be QAM symbols. We thus indeed get a MIDO code carrying
8 complex QAM symbols, with unitary encoding matrix yielding the cubic shaping
property. The factor  $\frac1{\sqrt{5}}$ is used to normalize  the minimum determinant to one.

%Note that $\sqrt{5}$ at the denominator is a normalization
%factor.

A straightforward calculation gives that the volume of the fundamental parallelotope
of this code is $5^4\cdot2^8$. At the same time, the minimum determinant of the code
is  1. If we now scale the code $\Cc_3$  with $(\frac1{5^4\cdot 2^8})^{1/16}$, the
resulting code lattice $\Cc_3 ^*=(\frac1{5^4\cdot 2^8})^{1/16}\cdot\Cc_3$ has a fundamental
parallelotope of volume $1$. We now see that the normalized minimum determinant of the lattice
$\Cc_3^*$ is
$$\left[\left(\frac{1}{5^4\cdot 2^8}\right)^{1/16}\right]^4=\frac1{20}.$$
Comparing this to Proposition
\ref{optimality}, we conclude that the normalized minimum determinant of the code $\Cc_3$
is very close to the optimum minimum determinant of orthogonally shaped MIDO codes.
The good performance of this code once again suggests that it is favorable for the code performance at low SNRs to maintain the cubic shaping.

Take again a codeword
\[
\left(
\begin{array}{rrrr}
x_0 & i\sigma(x_3)& ix_2 & i\sigma(x_1)\\
x_1 & \sigma(x_0) & ix_3 & i\sigma(x_2)\\
x_2 & \sigma(x_1) & x_0  & i\sigma(x_3)\\
x_3 & \sigma(x_2) & x_1  & \sigma(x_0)\\
\end{array}
\right)
\]
and multiply both the 3rd and 4th column by $\zeta_8^{-1}$, where
$\zeta_8=e^{2i\pi/8}$ is a primitive 8th root of unity. Then multiply the
3rd and 4th row this time by $\zeta_8$. Note that this of course brings
the matrix entries out of the algebra we started with, but will do this without
changing the determinant. We further note that we can use $\gamma=-i$ instead
of $\gamma=i$, since $-i$ is not a norm. The proof of this fact is similar to that of the non-norm element $i$ (cf. \ref{background}), and follows from the same argument of the transitivity of the norm. We have to show that
there cannot be an element with norm $-i$ over $\QQ(i,\sqrt{5})/\QQ(i)$. If there
were an element $a$ with $\Nc_{\QQ(\sqrt{5},i)/\QQ(i)}(a)=-i$, then $ia$ would have norm
\[
i^2\Nc_{\QQ(\sqrt{5},i)/\QQ(i)}(a)=i,
\]
a contradiction. Again for the case of $\Nc_{\QQ(\sqrt{5},i)/\QQ(i)}(a)=\gamma^2=-1$ we refer the reader to \cite[Section 8]{HLRV2}.

We now obtain for the codebook $\Cc_3$ consisting of matrices 
\begin{equation}\label{eq:MIDO3}
\left(
\begin{array}{rrrr}
x_0 & -i\sigma(x_3)& -\zeta_8x_2 & -\zeta_8\sigma(x_1)\\
x_1 & \sigma(x_0) & -\zeta_8x_3 & -\zeta_8\sigma(x_2)\\
\zeta_8x_2 & \zeta_8\sigma(x_1) & x_0  & -i\sigma(x_3)\\
\zeta_8x_3 & \zeta_8\sigma(x_2) & x_1  & \sigma(x_0)\\
\end{array}
\right).
\end{equation}
%%%%%korjaa
Let us denote by ${\bf c}_1$, ${\bf c}_2$, ${\bf c}_3$ and ${\bf c}_4$
the 4 columns of the above matrix.
It can be easily seen that the above manipulations result in having
columns 1 and 3, and 2 and 4 satisfying 
\[
{\bf c}_1^T{\bf c}_3=0,~{\bf c}_2^T{\bf c}_4=0
\]
without changing the shaping. This construction thus increases 
the ``orthogonality-likeness'' of the columns of the code without altering its other properties. Though this transformation does increase the number of zeroes 
in the $R$-matrix of the QR decomposition, it does not reduce the 
decoding complexity as defined. This is due to the fact that, albeit the above relations resemble the real inner product, the vectors $\mathbf{c}_i$ actually consist of complex elements. 

We now propose another puncturing, which focuses this time
on having orthonormal columns, in order to have  provable fast decodability.
Since $\QQ(\zeta_5,i)=\QQ(\zeta_{20})$, where $\zeta=\zeta_{20}=e^{2i\pi/20}$ is a
primitive 20th root of unity, we can alternatively take as basis for
$\QQ(i,\zeta_5)$ the set $\{1,\zeta,\zeta^2,\zeta^3\}$. An element
$x$ is then written as
\[
x=a+b\zeta+c\zeta^2+d\zeta^3,~a,b,c,d\in\QQ(i).
\]
We perform the following puncturing and restriction of coefficients. Take $x_0,x_1$ of the form
\[
a+ib\zeta+c\zeta^2+id\zeta^3,~a,b,c,d\in\ZZ
\]
so that $\sigma^2(x_0)={x_0}^*$, $\sigma^2(x_1)={x_1}^*$. For $x_2$
and $x_3$, take instead
\[
a(1+i)+b(1-i)\zeta+c(1+i)\zeta^2+d(1-i)\zeta^3,~a,b,c,d\in\ZZ
\]
to get this time $\sigma^2(x_2)=-{x_2}^*$, $\sigma^2(x_3)=-{x_3}^*$.
This results in a codebook $\Cc_2$ with codewords given by
\begin{equation}\label{eq:MIDO2}
X=
\left(
\begin{array}{rrrr}
x_0 & i\sigma(x_3)& -{x_2}^* & i{\sigma(x_1)^*}\\
x_1 & \sigma(x_0) & -{x_3}^* & -{\sigma(x_2)^*}\\
x_2 & \sigma(x_1) & {x_0}^*  & -{\sigma(x_3)^*}\\
x_3 & \sigma(x_2) & {x_1}^*  & {\sigma(x_0)^*}\\
\end{array}
\right).
\end{equation}
An easy computation shows that the 1st and 3rd row, resp. the
2nd and 4th row, are orthonormal.
By permuting the 2nd and 3rd rows and columns resp.,
we get
\begin{eqnarray}
&& X_{\Cc_2}(x_0,x_1,x_2,x_3)=\nonumber \\
&&\left(
\begin{array}{rrrr}
x_0 & -x_2^*  & i\sigma(x_3)& i{\sigma(x_1)}^*\\
x_2 &  x_0^*  & \sigma(x_1) & -{\sigma(x_3)}^*\\
x_1 & -{x_3 ^*} & \sigma(x_0) &  -{\sigma(x_2)}^*\\
x_3 &  {x_1^*}  & \sigma(x_2) &  {\sigma(x_0)}^*\\
\end{array}
\right)\label{eq:MIDO2-AF}
\end{eqnarray}
which clearly exhibits the Alamouti block structure of the code.

As previously for the $A_4$-code, a PAM vector $(g_1,\dots,g_{16})$ is  mapped into a $(4\times4)$ matrix
$$
\sum_{i=1}^{16} g_i B_i,
$$
where the basis matrices $B_i$ are
\begin{eqnarray*}
B_1&=&X_{\Cc_2}(x_0(1,0,0,0),0,0,0),\\
B_2&=&X_{\Cc_2}(x_0(0,1,0,0),0,0,0),\\
&\vdots&\\
B_{16}&=&X_{\Cc_2}(0,0,0,x_3(0,0,0,1)).\\
\end{eqnarray*}
Again a direct calculation gives
\[
\Re(\mathrm{Tr}(HB_i (HB_j)^{\dagger})=0
\]
for $1\leq i \leq 4$ and $5\leq j \leq 8$ and a complexity of $|S|^{12}$.

\subsection{The Srinath-Rajan (SR) code}
\label{SR-code}
So far, the best performing fast-decodable $4\times 2$ code has been the code based on stacked CIODs proposed in \cite{SR-fast}. The real and imaginary parts of the encoded symbols are separated in a careful way, so that when a rotated 4- or 16-QAM alphabet is used, the code has high coding gain. It is moreover conjectured that the code has the NVD property, but this has not been proved. Before rotating the constellation, the code is equivalent to transmitting four independent Alamouti blocks $A,B,C,D$:
$$
X_{\textrm{SR unrotated}}=\left(\begin{array}{cc}A&\zeta_8 B\\ \zeta_8C &D \end{array}\right),
$$
where  a primitive 8th root of unity $\zeta_8$  has been added in order to maximize the coding gain of the rotated code. Because the blocks are independent prior to rotation, the unrotated code does not have full diversity. For this reason, getting a proof for the possible NVD by using the theory developed in this paper does not seem  possible.  

If we ignore the constant $\zeta_8$, the code is exactly of the same form as the codes proposed in this paper (except possibly for the NVD), as clearly
$$
\left(\begin{array}{cc}A& B\\ C&D \end{array}\right)\in M_2(\mathbf{H}).
$$
Adding the constant $\zeta_8$ does not affect fast decodability, but helps to maximize the coding gain.
 
We have not tried whether it is possible to improve the coding gain of the  codes  proposed in this paper by using a suitable rotation. This may be seen as a reason for the small performance loss of the proposed codes compared to the rotated SR code. We did however try another type of optimization, namely using a spherical constellation instead of linearly dispersed constellation (cf. \ref{mindetorder}). The spherically shaped fast-decodable code outperforms  the SR code (see Section \ref{sec:simu} below) by a fraction of a dB. 

%**********************************************************************%

%cami tŠhŠn asti, paitti korjaa orto-like
\section{Simulation results of MIDO codes}
\label{sec:simu}

In Figure \ref{fig:simul}, we have plotted the block error rates of different MIDO 
codes at the rate 4 bpcu. All of the codes use the 2-PAM or 4-QAM alphabet, except for the spherically 
shaped $A_4$ code referred to as $NC\, (FD,A4,\textrm{spher.})$. This code  is constructed  by using a 6-PAM alphabet and then choosing the codewords with the smallest Frobenius norms, resulting in a codebook with  $2^{16}$ codewords. 

We can see that the punctured code 
$\Cc_2$ ($NC\, (FD,\textrm{punct.})$) does not perform too well due to its small (though non-vanishing) coding gain. 
The other new codes, for their part, perform more or less equally to the Biglieri-Hong-Viterbo (BHV) 
code. The $A_4$ code ($NC\, (FD,A4)$) is slightly beaten by the BHV code at low-moderate SNRs, but will eventually 
outperform it starting from 20 dB, thanks to its full diversity. The shaped code ($NC\, (\textrm{shaped})$), which is not 
fast-decodable, outperforms the BHV code starting from 16 dB. The Srinath-Rajan (SR) code with a rotated 4-QAM constellation wins the BHV code by a fraction of a dB. The spherically shaped $A_4$ 
outperforms the BHV and SR codes by roughly 0.5 dB, and performs slightly better compared to the best 
previously known MIDO code IA-MAX \cite{asykonstru}. The code IA-MAX is constructed from a 
certain maximal order, and has higher decoding complexity. It is added here for the sake of 
completeness in comparison.

Let us point out that we have not optimized any of the proposed codes by \emph{e.g.} rotating the constellation. Just out of interest, we  simulated the unrotated  SR code, and the performance got somewhat worse than that of the $A_4$ code. Hence, we also expect some improvement in the performance of our codes, when an optimal rotation is used.

\begin{figure}%[htp!]
\includegraphics[scale=0.5]{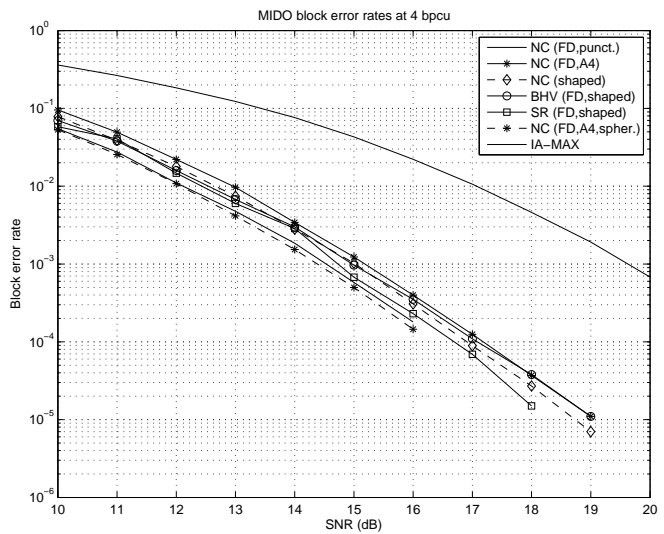}
\caption{Comparison among different MIDO codes at rate 4 bpcu.}
\label{fig:simul}
\end{figure}

We can also use the  maximal order of the $A_4$ code algebra, which will result in similar 
performance as the IA-MAX and spherically shaped $A_4$ code. While the maximal order codes 
are not fast-decodable, the spherically shaped $A_4$ code still uses the same linear dispersion 
matrices and hence admits fast decodability. However, an extra step is required to check that 
the decoded word really belongs to the codebook. For a detailed description of the required changes in a sphere decoder, see \cite{maxdecoding}. As a conclusion, sticking to linear dispersion 
and natural orders causes a penalty of about 0.5 dB in the BLER performance. On the other hand, 
it seems that the requirement of fast-decodability itself does not cause any performance loss. 
This is hardly a surprise, as the proposed constructions are nothing but orders of cyclic 
division algebras, which have been shown to have excellent performance \cite{HLRV2,CK,asykonstru}.

%******************************************************************************%
%
% 6Tx
%
%******************************************************************************%

\section{Fast-decodable codes for the $6\times 3$ and $6\times 2$ channels}
\label{sec:codes6}

Let us now extend our code constructions to six transmit antennas. While this paper mainly deals 
with MIDO codes, \emph{i.e.}, codes for two receivers, here we also consider the case of three receivers. 
The reason for this is that the embedding (\ref{psikuvaus})
$$
\psi:\mathcal{A}\hookrightarrow\textrm{diag}(M_{n/2}(\mathbf{H})^m)
$$
into to a matrix ring of the Hamiltonian quaternions naturally yields codes with dimension rate $R_1=n$,  
which is also the number of Tx antennas. Thus, for six transmitters we have $R_1=6$, which is ideal for 
reception with three antennas. From this, we can construct a code suitable for two receivers ($R_1=4$) by puncturing. The so-called \emph{smart puncturing} \cite{spm,asykonstru} will be applied in order to further reduce the decoding complexity, while maintaining a low peak-to-mean power ratio (PAPR).

\subsection{Construction for the $6\times 3$ channel}
We build our $(6\times 6)$ code matrix analogously to the $(4\times 4)$ case (cf. Subsection \ref{family}). To this end, we consider the index-six cyclic algebra
$$
\mathcal{A}=(\mathbb{Q}(\zeta_7)/\mathbb{Q},\sigma:\zeta_7\mapsto \zeta_7^3,-3/4)
$$
built upon the 7th cyclotomic field. Since -3 is inert ($(3 \ \textrm{mod}\ 7)$ generates the whole group $\mathbb{Z}_7^*$), the element $\gamma=-3/4$ is a non-norm element and $\mathcal{A}$ is a division algebra.  As the center $\mathbb{Q}$ is totally real and only has one infinite place which is ramified, we have an embededding $\mathcal{A}\hookrightarrow M_{3}(\mathbf{H})$.

Let us now build the embedded code matrix more explicitly. We start by noting that
$$
\sigma^3(x)=x^*
$$
for all $x\in\mathbb{Q}(\zeta_7)$, and hence, taking into account that  $\sigma(x^*)=\sigma(x)^*$, we get
$$
\sigma^4(x)=\sigma(x)^*, \quad \sigma^5(x)=\sigma^2(x)^*.
$$
We can again start with the left regular representation   and perform some simple manipulation on the resulting matrix: first, we swap the 2nd and the 4th row, and the 3rd and the 5th row. After this, we swap the 3rd and the 4th row. Next, we do the same for the columns. Let us denote this intermediate form by $X'$. Then we balance the effect of $\gamma$ to get a more unified energy distribution among the antennas. This can be done by conjugating the matrix $X'$ by the matrix
$$
P=\textrm{diag}(r,r^2,r,r^2,r,r^2),
$$
where $r=\sqrt{|\gamma|}$. Finally, we do the exchange $x_3\leftrightarrow x_1$ and $x_4\leftrightarrow x_2$, followed by $x_2\leftrightarrow x_3$.
The final form of the code matrix now becomes
\begin{equation}\label{6x3}
X=PX'P^{-1}=\left(\begin{array}{ccc}A&B&C\end{array}\right),
\end{equation}
where each
\begin{equation}
A=\left(\begin{array}{rr}
x_0 & -rx_1^*  \\
rx_1 & x_0^*  \\
x_2 & -rx_3^*  \\
rx_3  & x_2^*  \\
x_4 & -rx_5^*  \\
rx_5 & x_4^*  \\
\end{array}\right),
\end{equation}
\begin{equation}
B=\left(\begin{array}{rr}
 -r^2\sigma(x_5)& -r\sigma(x_4)^*   \\
 r\sigma(x_4) & -r^2\sigma(x_5)^* \\
 \sigma(x_0) & -r\sigma(x_1)^*  \\
 r\sigma(x_1) & \sigma(x_0)^*   \\
 \sigma(x_2) &  -r\sigma(x_3)^*  \\
 r\sigma(x_3) & \sigma(x_2)^*  \\
\end{array}\right),
\end{equation}
and
\begin{equation}
C=\left(\begin{array}{rr}
 -r^2\sigma^2(x_3) & -r\sigma^2(x_2)^*\\
 r\sigma^2(x_2) & -r^2\sigma^2(x_3)^*\\
 -r^2\sigma^2(x_4) & -r\sigma^2(x_5)^*\\
 r\sigma^2(x_5) & -r^2\sigma^2(x_4)^*\\
 \sigma^2(x_0) & -r\sigma^2(x_1)^*\\
 r\sigma^2(x_1) & \sigma^2(x_0)^*\\
\end{array}\right)
\end{equation}
consist of three  Alamouti blocks.

The encoding can be performed similarly as in the $4\times 2$ case. Let us denote the 36 basis matrices by $$B_1=B_1(x_0(1,0,0,0,0,0),0,0,0,0,0),$$
$\hspace{2cm}\vdots$  $$B_{2}=B_{2}(x_0(0,1,0,0,0,0),0,0,0,0,0),$$ $$B_{36}=B_{36}(0,0,0,0,0,x_5(0,0,0,0,0,1)).$$ We then form a finite code by setting 
$$
\Cc_{6\times 3}=\{\sum_{i=1}^{36}g_iB_i\ |\ g_i\in \mathcal{G} \},
$$
where $\mathcal{G}\subseteq\Z$ is, for instance, a $Q$-PAM alphabet. 

%The above row/column operations, element exchanges and conjugation are irrelevant in terms of the coding gain, but provide us with reduced decoding complexity as we shall see next.

\subsection{Decoding}
Let us now consider the sphere decoding process as described in \ref{sec:FD} for the code (\ref{6x3}). Following the above notation, we notice that the code lattice has six basis matrices $B_1,\ldots,B_6$ of the form
$$
\left(\begin{array}{cccccc}
x_0&&&&&\\
&x_0^*&&&&\\
&&\sigma(x_0)&&&\\
&&&\sigma(x_0)^*&&\\
&&&&\sigma^2(x_0)&\\
&&&&&\sigma^2(x_0)^*\\
\end{array}\right),
$$
and six basis matrices $B_7,\ldots,B_{12}$ of the form
$$
\left(\begin{array}{ccc}
A'&0&0\\
0&B'&0\\
0&0&C'\\
\end{array}\right),
$$
where
$$A'=\left(\begin{array}{cc}0&-rx_1^*\\
rx_1&0\\
\end{array}\right),$$
$$B'=
\left(\begin{array}{cc}
0&-r\sigma(x_1)^*\\
r\sigma(x_1)&0\\
\end{array}\right),$$
and
$$C'=
\left(\begin{array}{cc}0&-r\sigma^2(x_1)^*\\
r\sigma^2(x_1)&0\\
\end{array}\right).
$$
A straightforward calculation shows that
$$
\Re(\textrm{Tr}(HB_i(HB_j)^\dag))=0
$$
for  $1\leq i\leq 6,\ 7\leq j\leq 12$ and any channel matrix $H$. Hence, the $(36\times 36)$ $R$-matrix of the QR decomposition has a $(6\times 6)$ zero block in the corresponding position, and the $(12\times 12)$ upper left corner of $R$ looks like
$$
\left(\begin{array}{cc}
R^{1,1} & 0\\
0 & R^{2,2}\\
\end{array}\right),
$$
where the blocks $R^{i,i}$ are $(6\times 6)$ matrices. From this we see that the symbols $g_1,\ldots,g_6$ can be decoded independently of the symbols $g_7,\ldots,g_{12}$, resulting in complexity $2|S|^{30}$ instead of the full complexity $|S|^{36}$. Further reductions are possible by a change of basis, similarly as in the $4\times 2$ case. By forming the basis of elements half of which are real and the other half  purely imaginary (cf. \ref{family}), we get more zeros in the $R$ matrix. In that case we again have, for any channel matrix $H$, that
$$
\Re(\textrm{Tr}(HB_i(HB_j)^\dag))=0
$$
 for $1\leq i\leq 6,\ 7\leq j\leq 12$, but further also get
 $$
\Re(\textrm{Tr}(HB_i(HB_j)^\dag))=0
$$
for  $1\leq i\leq 3,\ 4\leq j\leq 6$ and  $7\leq i\leq 9,\ 10\leq j\leq 12$, resulting in complexity $4|S|^{27}$.

 \subsection{Construction for the $6\times 2$ channel by puncturing}
In order to construct a $6\times 2$ MIDO code, we will next consider a punctured version of the above code. The puncturing affects the shape of the code lattice, so different puncturing will give a different lattice and whence also different performance. One obvious option is to keep an eye on the Gram matrix of the lattice -- the closer it is to a (scaled) identity matrix, the better the shape.  A smart puncturing may also aid the decoding process, namely we may puncture the basis matrices that cause nonorthogonality. On the other hand, it is not a good idea to puncture all six basis matrices corresponding to one of the elements $x_i$ in (\ref{6x3}), because this will cause zeros in the encoding matrix and hence increase the PAPR.

Here we provide just one possible puncturing, to give the reader an idea as to how one may go about it.
Let us denote the basis matrices as in the previous section by $B_1,\ldots,B_{36}$.
We puncture the following basis matrices
$$\textrm{in } x_2:\quad B_{13},B_{14},B_{15},$$
$$\textrm{in } x_3:\quad B_{19},B_{20},B_{21},$$
$$\textrm{in } x_4:\quad B_{25},B_{26},B_{27},$$
$$\textrm{in } x_5:\quad B_{31},B_{32},B_{33}.$$
The resulting code will still have the same orthogonality relations as the original code, but  will only have 24 basis elements giving us decoding complexity $4|S|^{15}$.

 \section{Further generalizations via conjugations of the left-regular representation}\label{sec:furthergen}
 As already pointed out, we can always embed a division algebra into a matrix ring of the Hamiltonian quaternions, provided that the center is totally real and all of its infinite places ramify. For all such division algebras, we have that $\sigma^{n_t/2}(x)=x^*,\, \sigma^{j+n_t/2}(x)=\sigma^{j}(x)^*$, and $\gamma<0$. The embedding 
$$
\psi:\mathcal{A}\hookrightarrow\textrm{diag}(M_{n/2}(\mathbf{H})^m),
$$ 
however, will only give us the existence of a fast-decodable code with dimension rate $n=n_t$. 
%--------
%--------

In what follows, we are going to show how to overcome the problem of the implicit nature of the map $\psi$. Once we have constructed a CDA $\mathcal{A}=(E/\mathbb{Q},\sigma,\gamma)$ of the required form, the explicit map $\psi:\mathcal{A}\rightarrow M_{n_t/2}(\mathbf{H})$ is given as follows. 
\begin{proposition}\label{thm:general}
 Let $X$ denote the left regular representation matrix of an element $a=x_0+ux_1+\cdots+u^{n_t-1}x_{n_t-1}\in\mathcal{A}$. Then
$$
\psi(X)=BPX(BP)^{-1}\in M_{n_t/2}(\mathbf{H}),
$$
where the elements $P(i,j), \,1\leq i, j\leq n_t,$ of the permutation matrix $P$ are
$$
P(i,j)=\left\{\begin{array}{clc} 1, & \ \textrm{if}\quad  2\not|\,\, i\quad \textrm{and}\quad j=\frac{i+1}2,\\
1, & \  \textrm{if} \quad  2\, | \, \, i \quad \textrm{and} \quad j=\frac{i+n_t}2,\\
0, &  \ \textrm{otherwise}
 \end{array}\right.
$$ 
and
$$
B=\textrm{diag}(\sqrt{|\gamma|},|\gamma|,\ldots,\sqrt{|\gamma|},|\gamma|)
$$
is the energy balance matrix. 
\end{proposition}
\begin{proof} Let us first consider the columns of $X$, and denote $X=(1, \sigma, \ldots, \sigma^{n_t-1})$ to represent the fact that the first column is mapped by the identity element, the second is mapped by $\sigma$, etc. In order to get the required Alamouti block form, we need to reorder the columns as
$$
(1,\sigma^{n_t/2},\sigma^2,\sigma^{2+n_t/2},\ldots,\sigma^{n_t/2-1}\sigma^{n_t-1}),
$$
so that $\sigma^j$ is followed by its conjugate for all $j$. This is done exactly by post-multiplying $X$ by $P^{-1}$.

Next we have to rearrange the rows. Notice first that, by ignoring $\gamma$, the rows of $XP^{-1}$ look like
$$\left(\begin{array}{ccccc}
a_1&b_1^*&\ldots& a_{n_t/2}&b_{n_t/2}^* \\
c_1&d_1^*&\ldots& c_{n_t/2}&d_{n_t/2}^* \\
\vdots&\vdots&&\vdots&\vdots\\
s_1&t_1^*&\ldots&s_{n_t/2}&t_{n_t/2}^* \\
\hline 
b_1&a_1^*&\ldots& b_{n_t/2}&a_{n_t/2}^* \\
d_1&c_1^*&\ldots& d_{n_t/2}&c_{n_t/2}^* \\
\vdots&\vdots&&\vdots&\vdots\\
t_1&s_1^*&\ldots&t_{n_t/2}&s_{n_t/2}^*\\
\end{array}\right)$$ 
where the horizontal line divides the matrix in two parts each having $n_t/2$ rows. We easily see that the Alamouti block form can be achieved by pairing the rows as
$$
(1,n_t/2+1),(2,n_t/2+2),\ldots,(n_t/2,n_t).
$$
This is done by pre-multiplying $XP^{-1}$ by $P$, \emph{i.e.}, we conjugate $X$ by $P$. As the last step, we should take care of the effect of $\gamma$. By conjugating $PXP^{-1}$ further by $B=\diag(\sqrt{|\gamma|},|\gamma|,\ldots,\sqrt{|\gamma|},|\gamma|)$, the elements $\gamma$ will appear in each $(2\times 2)$ block of the matrix as follows:
$$
\left(\begin{array}{cc}
(\pm)\sqrt{|\gamma|}&(\pm) |\gamma|\\
(\pm)|\gamma|&(\pm)\sqrt{|\gamma|}\\ 
\end{array}\right).
$$
In addition, the plus and minus signs are automatically rearranged by this conjugation so that the resulting matrix consists of Alamouti blocks. \end{proof}

\begin{remark}
After Proposition \ref{thm:general}, we can algebraically optimize the normalized minimum determinant. Namely, the resulting parallelotope will be  exactly that given by Proposition \ref{thm:conjugation}. Notice that this was not the case before the conjugation, for while the conjugation  does not affect the non-normalized minimum determinant, it does affect the measure of the fundamental parallelotope and hence the normalized minimum determinant!
\end{remark}

%----------
%-------
Now that we have an explicit form of the mapping $\psi$, the fast-decodability property can be seen as follows: with $\mathbb{Q}$ as the center ($m=1$), the $R$-matrix of the QR decomposition of the matrix $B$ (cf. \ref{sec:FD}) will consist of  $(n\times n)$  blocks  $R^{i,j},\ 1\leq i,j \leq n$, where 
\begin{equation} \label{n-reduction}
R^{1,2}=R^{3,4}=\cdots=R^{n-1,n}=\mathbf{0}_{n\times n}
\end{equation} 
and the diagonal blocks $R^{i,i},\ 1\leq i \leq n$, are block-diagonal: 
\begin{equation}\label{n/2-reduction}
R^{i,i}=\left(\begin{array}{cc}
P^{1,1}&0\\
0& P^{2,2}\\
\end{array}\right)_{n\times n}.
\end{equation}
The zero blocks (\ref{n-reduction}) result from the Alamouti block structure and offer us a reduction of $n$ real dimensions. The diagonal block structure (\ref{n/2-reduction}) is due to the fact that when we construct the algebra upon a complex multiplication field, we can always choose a basis in which half of the elements are real and the other half purely imaginary. This, for its part, provides us with further reduction by $\frac n2$ dimensions.  Hence,  the decoding complexity will be of order 
$$\leq |S|^{n_t^2-n_t-\frac {n_t}2}=|S|^{n_t^2-\frac {3n_t}2},$$ 
where the factor ${n_t}^2$ is the exhaustive search complexity. 

By puncturing, we obtain fast-decodable codes suitable for any number of receivers. The complexity of the punctured code is at most $$|S|^{n_tR_1-\frac {3n_t}2},$$ where $R_1\leq n_t$ is the dimension rate. For $n_r=2$, we get a complexity reduction of $\frac{4n_t-2.5n_t}{4n_t}=37.5\%$ as promised. However, this may require a non-integral basis, and hence cause performance loss compared to an integral basis. With an integral basis, we get a reduction of $\frac{4n_t-n_t}{4n_t}=25\%$ while guaranteeing a high coding gain. 

In Table \ref{tab:comp} we have summarized the complexities for $n_t=4, 6, 8$ and $2\leq n_r\leq \frac{n_t}2$ as an example.

 \begin{table}[thp!]\caption{Complexities of the proposed fast-decodable codes.}\label{tab:comp}
\begin{center}
\begin{tabular}{|c|c|c|c|}
\hline
$n_t\times n_r$ & $R_1$ & $n_tR_1-\frac{3n_t}2$ & Comp.reduction$/n_tR_1$\\
\hline
$4\times 2$ & $4$ & $10$ & $37.50\%$\\
$6\times 3$ & $6$ & $27$ & $25.00\%$\\
$6\times 2$ & $4$ & $15$ & $37.50\%$\\
$8\times 4$ & $8$ & $52$ & $18.75\%$\\
$8\times 3$ & $6$ & $36$ & $25.00\%$\\
$8\times 2$ & $4$ & $20$ & $37.50\%$\\
\vdots & \vdots & \vdots & \vdots \\
\hline
\end{tabular}
\end{center}
\end{table}

%*********************************************************************%
%
% CONCLUSION
%
%*********************************************************************%

\section{Conclusions}
\label{sec:conc}
In this paper, fast-decodable asymmetric lattice space-time codes were studied, proposing one possible generalization of the Alamouti code and the quasi-orthogonal codes to any even number of transmit antennas $n_t$ and for any dimension rate $R_1\leq n_t$. The codes allow linear ML processing with \emph{e.g.} a sphere decoder for any number of receivers $\geq R_1/2$, but with  lower dimensionality (less variables per linear equation). It was explicitly shown how such novel constructions follow from general algebraic principles by embedding a division algebra into a matrix ring $M_k(\Hv)$ of the  Hamiltonian quaternions. All this is in strong contrast to the previous \emph{ad hoc} constructions of fast-decodable codes that have been specific to a certain number of antennas and lacking an obvious  generalization. The proposed codes furthermore enjoy the NVD property, a property that no other fast-decodable MIDO code found in the literature has been proved to have.

We mainly considered the $4\times 2$ MIDO case suitable for DVB-NGH, but also provided constructions for the $6\times 2$ and $6\times 3$ cases. The explicit embeddings obtained in these situations were shown to be fully generalizable to any even number of Tx antennas.
Simulations were presented to show that the performance of the proposed codes is comparable to the best known MIDO codes. The achieved complexity reduction up to 37.5\% is also among the best known for the MIDO channel.

In addition, a complete solution to the discriminant minimization problem for division algebras with arbitrary centers was given. As an application a normalized minimum determinant bound for code lattices in $M_k(\Hv)$ was derived from the algebraic  results.

%*********************************************************************%
%
% ACK
%
%*********************************************************************%

\section*{Acknowledgments}

The research of R. Vehkalahti and C. Hollanti is supported by the Emil Aaltonen
Foundation's Young Researcher's Project, and by the Academy of Finland 
grant \#131745.

This research was partly carried out during the visit of C. Hollanti at the Nanyang Technological University, Singapore. The research of F. Oggier is supported in part by the Singapore National
Research Foundation under Research Grant NRF-RF2009-07 and NRF-CRP2-2007-03,
and in part by the Nanyang Technological University under Research
Grant M58110049 and M58110070.

%******************************************************************************%
%
% BIBLIO
%
%******************************************************************************%

%\bibliographystyle{ieee}	% (uses file "plain.bst")
%\bibliography{myrefs_MIDO-IT}		% expects file "myrefs.bib

%***************************************************************************%
%
% APPENDIX
%
%***************************************************************************%

\section*{Appendix}

In this Appendix we are going to present some basic results from the theory of central simple algebras and in particular from the theory of Hasse-invariants. These results are needed only in Section \ref{bounds}.

For a quick introduction we refer the reader to \cite{HLRV2} and \cite{Roopenkirja} where similar optimization has been done.

Let us consider a $K$-central division algebra of index $n$.
Then attached to each pair $(\A, P)$, where $P$ is a prime of $K$,  is a positive rational number $h_P={a/m_P}$, the so-called \emph{Hasse-invariant} of $\A$ at $P$.
The Hasse invariants of $\A$ fulfill the following.
When $P$ is a prime ideal of $K$, then
$$
h_P=\frac{a}{m_P},  \quad  0\leq a < m_P \leq n,\,\, (a, m_P)=1,
$$
 when $P$ is infinite and real, then
$$
h_P=1/2 \ \mathrm{ or } \ h_P=0,
$$
and when  $P$ is infinite and complex, then $$ h_P=0.$$
The number $m_P$ is called the \emph{local index} at prime $P$ (see Section \ref{mainmapsection}). We say that the algebra $\D$ is \emph{ramified} at the prime $P$, if
$h_P\neq 0$.
The Hasse invariants define the algebraic structure of a division algebra and in particular the discriminant of the algebra.
\begin{proposition}\label{collection}
Assume  that  $P_1,\ldots,
P_s $ are a set of finite  prime ideals of $\mathcal{O}_K$ and $P_{s+1},\dots, P_n$ are a set of real primes.

Assume further that a sequence of rational
numbers
$$
\frac{a_1}{m_{P_1}},\ldots , \frac{a_s}{m_{P_s}}, \frac{a_{s+1}}{m_{P_{s+1}}},\ldots, \frac{a_{n}}{m_{P_{n}}},
$$
subject to the restriction that when $i>s$, $a_i/m_{P_i}=1/2$,
satisfies
$$
\sum_{i=1} ^{n} \frac{a_i}{m_{P_i}} \equiv 0 \qquad  (
\textrm{mod}\quad 1 ),
$$
$1\leq a_i \leq m_{P_i}$, and $(a_i , m_{P_i})=1$.

Then there exist a $K$-central division algebra $\mathcal{A}$ that has local
indices $m_{P_i}$ and the least common multiple (LCM) of the
numbers $\{m_{P_i}\}$ as an index.

If $\Lambda$ is a maximal $\mathcal{O}_K$-order in $\mathcal{A}$, then
 the discriminant of $\Lambda$ is
$$
 d(\Lambda/ \mathcal{O}_K )= \prod_{i=1} ^s
P_i^{(m_{P_i}-1)\frac{[\A:K]}{m_{P_i}}}.
$$
\end{proposition}

We have  the following two general results.

\begin{thm}[\cite{HLRV2}]\label{bound}
Let us suppose that we have a number field $K$ and an integer $n$, where  $4\mid n$ or $2\nmid n$. If $P_1\leq P_2$ is a pair of smallest primes in $\OO_K$, then there exists a $K$-central division algebra of index $n$ having a maximal order with the $\OO_K$-discriminant
$$
(P_1 P_2)^{n(n-1)}.
$$
This is the smallest possible discriminant for an order inside any $K$-central division algebra of index $n$.
\end{thm}

The following result is from \cite{Roopenkirja}, but is presented here for he first time in an article.

\begin{thm} [\cite{Roopenkirja}]\label{infinitebound}
Let $\mathcal{A}$ be a $K$-central division algebra of index
$2k=n$, where $k$ and $2$ are relatively prime and let $P_1\leq P_2$ be a pair of smallest primes in $\OO_K$.

If $K$ has at least two real primes, then there exists a $K$-central division algebra of index $n$ having a maximal order with the discriminant
$$
(P_1P_2)^{k(k-1)}.
$$

If $K$ has only one real prime $P_{\infty}$,  then there exists a $K$-central division algebra of index $n$ having a maximal order with the discriminant
$$
P_1^{n(n-1)}P_2^{k(k-1)}.
$$
This is  the smallest possible discriminant of all orders of index $n$ division algebras with center $K$.
\end{thm}

We have now given completely general discriminant bounds for any center and for any index $n$. 

\begin{proposition}
Let $\mathcal{A}$ be a  $K$-central division algebra of index
$n$, $2 \mid n$, where $K$ is a totally real number field, and let $P_1\leq P_2$ be a pair of smallest primes in $K$. Let us suppose that all the infinite primes are ramified in $\mathcal{A}$.

If  $2\mid \mid n$ and  $ 2\mid [K:\Q]$, then the minimal discriminant of
$\mathcal{A}$ is
$$
(P_1P_2)^{k(k-1)}.
$$

If $4\mid n$ then the minimal discriminant of
$\mathcal{A}$ is
$$
(P_1P_2)^{n(n-1)}.
$$
If $2\mid\mid n$ and  $2\nmid [K:\Q]$,  then the minimal discriminant of
$\mathcal{A}$ is
$$
P_1^{n(n-1)}P_2^{k(k-1)}.
$$
\end{proposition}

\begin{proof}
In the proofs of   Theorems \ref{infinitebound} and \ref{bound}  the general strategy was to choose a set of H-invariants that will yield an index $n$ division algebra (see Theorem \ref{collection}) and then prove that our choice was the  best possible. We will use the same strategy here, but the difference is that we can do  the optimization over division algebras that are totally ramified at infinite primes.

The assumption of ramified infinite primes always gives us $m$ non-trivial Hasse invariants $\{h_{P_1},\dots, h_{P_m}\}$, where $h_{P_i}=\frac{1}{2}$ and $P_i$ are all the infinite primes in $K$.

The Hasse-invariants at infinite places do not contribute anything
on discriminant of the division algebra. If we have an index $n$ division algebra, the contribution of a Hasse-invariant $h_{P}=\frac{s}{m_P}$, where $m_P$ is the local index at finite prime $P$,   to     the $\OO_K$-discriminant is $P^{(m_{P}-1)\frac{n}{m_{P}}}$. Therefore in most cases we can simply prove the minimality  of the corresponding discriminant by showing that,  despite the extra ramification at infinite primes, we can choose a set of  Hasse-invariants that will give  us an index $n$ division algebra with a discriminant reaching the bound  \ref{infinitebound} or \ref{bound}.

\begin{table}[thp!]\caption{}\label{table3}
\begin{center} \begin{tabular}{|l|l|l|l|l|l|}
\hline
 index & [K:Q]& H-invariants at finite places   \\ \hline
 odd &    & -                           \\
 4k& odd     &   $ h_{P_1}=\frac{1}{4k}$, $h_{P_2}=\frac{2k-1}{4k}$              \\
 4k&  even    & $ h_{P_1}=\frac{1}{4k}$, $h_{P_2}=\frac{k-1}{4k}$                    \\
 2k, $2\nmid k$& even & $h_{P_1}=\frac{1}{k}$, $h_{P_2}=\frac{k-1}{k}$                         \\
 2k, $2\nmid k$& odd    &  $h_{P_1}=\frac{k-2}{2k}$, $h_{P_2}=\frac{1}{k}$                 \\
 \hline
\end{tabular}
\end{center}
\end{table}

In  Table \ref{table3}  we have  collected the Hasse-invariants (at finite places) of the algebras  we claim to be optimal. In addition to what is said in the table about the H-invariants at the finite places, we suppose that each of these algebras have H-invariants $\frac{1}{2}$ at all the infinite primes. By a direct calculation we can see that in each case we get a division algebra of index $n$ with all the infinite primes ramified. This will take care of the first two claims of the proposition. In the first case, where $2\mid \mid n$ and  $ 2\mid [K:\Q]$,  the division algebra given in the table will reach the claimed bound which coincides with the general bound in \ref{infinitebound}. In the case $4\mid n$  the algebras given in Table \ref{table3}  reach the bound \ref{infinitebound} and we are done with the second claim.

We are left with the case, where    $2\mid\mid n$ and   $2\nmid [K:\Q]=m$. In this case the problem is that while the  sum of the $m-1$ first infinite  Hasse-invariants is an integer, there is still one extra  infinite H-invariant $h_{P_m}=\frac{1}{2}$ we have to take care of. Therefore we are forced to choose Hasse-invariants $h_{P_1}=\frac{k-2}{2k}$  and $h_{P_2}=\frac{1}{k}$ for the finite places. The proof that this set  of Hasse-invariants will give us the optimal discriminant is verbatim the same as it is for the case where the center has exactly one real place. This case was dealt in the proof  of Proposition   \ref{infinitebound} and we refer the reader to \cite{Roopenkirja}.

\end{proof}

\end{document}